%% file: APAFA_vs3.tex
\documentclass[12pt]{article}
\usepackage{amsmath}
\usepackage{graphicx}
\usepackage{tikz}
\usepackage{hyperref}
\hypersetup{colorlinks,allcolors=black}
\usepackage{enumerate}
\usepackage{dsfont}
\usepackage{xcolor} 
\usepackage[authoryear]{natbib}
\usepackage{url}
\usepackage{tikz}
\usetikzlibrary{positioning}
\usepackage{amsmath}
\usepackage{amsmath,amsthm,amssymb}
\usepackage{subfigure}
\newcommand{\blind}{1}
\newcommand{\R}{\mathds{R}}

\newcommand{\Tr}{\mbox{Tr}}
\newcommand{\model}{APAFA}
\newcommand{\orth}{\mathcal{O}}

\newcommand{\change}[1]{\textcolor{black}{#1}}

\theoremstyle{plain}
\newtheorem{definition}{Definition}
\newtheorem{theorem}{Theorem}
\newtheorem{lemma}{Lemma}
\newtheorem{corollary}{Corollary}
\theoremstyle{definition}

\theoremstyle{definition} 

\begin{document}

\def\spacingset#1{\renewcommand{\baselinestretch}%
{#1}\small\normalsize} \spacingset{1}


\if1\blind
{
  \title{\bf Adaptive Partition Factor Analysis}
  \author{Elena Bortolato\thanks{E.B. acknowledges support from the European Union, under the ERC grant agreement ID: 864863.}\hspace{.2cm}\\
    Department of Business and Economics, Universitat Pompeu Fabra,\\  
Data Science Center, Barcelona School of Economics\\
    and \\
   Antonio Canale \\
    Department of Statistical Sciences, University of Padova}
  \maketitle
} \fi

\if0\blind
{
  \bigskip
  \bigskip
  \bigskip
  \begin{center}
    {\LARGE\bf Adaptive partition Factor Analysis}
\end{center}
  \medskip
} \fi

\bigskip
\begin{abstract}	Factor Analysis has traditionally been utilized across diverse disciplines to extrapolate latent traits that influence the behavior of multivariate observed variables. Historically, the focus has been on analyzing data from a single study, neglecting the potential study-specific variations present in data from multiple studies. Multi-study factor analysis has emerged as a recent methodological advancement that addresses this gap by distinguishing between latent traits shared across studies and study-specific components arising from artifactual or population-specific sources of variation. 
	In this paper, we extend the current \change{Bayesian} methodologies by introducing novel shrinkage priors for the latent factors, thereby accommodating a broader spectrum of scenarios---from the absence of study-specific latent factors to models in which factors pertain only to small subgroups nested within or shared between the studies. For the proposed construction we provide conditions for identifiability of factor loadings and guidelines to perform straightforward posterior computation via Gibbs sampling.
	Through comprehensive simulation studies, we demonstrate that our proposed method exhibits competing performance across a variety of scenarios compared to existing methods, yet providing richer insights. The practical benefits of our approach are further illustrated through applications to bird species co-occurrence data and ovarian cancer gene expression data.
\end{abstract}

\noindent%
\vfill

\newpage
\spacingset{1.9} 
\section{Introduction}
\label{sec:intro}

In numerous scientific and socio-economic domains, the collection of high-dimensional data has become ubiquitous. Instances of such data collection include customer preferences in recommender system applications, single-cell experiments in genomics, high-frequency trading in finance, and extensive sampling campaigns in ecology. In these contexts,  joint analysis of data originating from diverse sources, studies, or technologies is increasingly critical for enhancing the accuracy of conclusions and ensuring the reproducibility of results \citep{national2019reproducibility}. 

{In these contexts, factor analysis is a standard technique for identifying common patterns of variation in multivariate data and relating it to hidden causes, the so-called latent factors\label{change:goal}. 
However, a key challenge arises when the latent factors identified across different studies combine both shared factors (common across studies) and factors unique to individual studies.
This issue is particularly pronounced in our motivating contexts of animal co-occurrence studies and genomics. 
For instance, species occurrence data collected across spatial-temporal surveys often exhibit site-specific deviations despite shared environmental conditions \citep{ovaskainen2017make}. Similarly, high-throughput genomic experiments are prone to technical variability, where artifacts from experimental conditions may dominate biological signals \citep{irizarry2003exploration, shi2006microarray, vallejos2017normalizing, hicks2018missing}. While preprocessing adjusts for some biases, residual heterogeneity in covariance structures often persists across groups, challenging homogeneity assumption \citep{avalos2022heterogeneous}.}

Motivated by these problems, \citet{devito1}  introduced a novel methodological tool, namely multi-study factor analysis (MSFA), which extends the classical factor analysis to jointly analyze data from multiple studies. MSFA aims to separate the signal shared across multiple studies from the study-specific components arising from artifactual or population-specific sources of variation.  {Let ${y}_{is} \in \mathbb{R}^p$ denote the $p$-dimensional observation vector for the $i$-th subject ($i = 1, \dots, n_s$) in the $s$-th study ($s = 1, \dots, S$), where each study $s$ has $n_s$ subjects and the total sample size is $n = \sum_{s=1}^S n_s$.}
In the MSFA framework, one assumes the following representation:
\begin{equation}
	y_{is} =  \Lambda \eta_{is} + \Gamma_s {\varphi}_{is} + \epsilon_{is}
	\label{eq:msfa} \end{equation}
where $\epsilon_{is}$ represents a zero-mean idiosyncratic error term
and $\eta_{is}$ is a $d$-dimensional set of shared latent factors, with $d \ll p$ and with $\Lambda$ the corresponding factor loading matrix  of dimension  $p \times d$. Similarly,  $\Gamma_s$ is a study-specific factor loading matrix of dimension $p\times k_s$,  with $k_s \ll p$, and ${\varphi}_{is}$ its corresponding latent factors. {Notably,  each subject belongs to a single study and, in fact,  the index $s$ in the notation $y_{is}$ may be redundant as also discussed later in Section 2. Figure 1 reports a pictorial representation of the data matrix $Y$ with $n$ rows defined as $Y = (y_{11}^\top, \dots, y_{n_SS}^\top)^\top$.} 

To address the high-$p$ small-$n$ problem typical of high-throughput biological data, \citet{devito2} propose a Bayesian generalization of the MSFA model that {naturally provides \label{change:necessary} regularization } through the multiplicative gamma prior of \citet{bhattacharya2011sparse}. Posterior sampling is performed via Markov Chain Monte Carlo (MCMC) methods, enabling flexible estimation without constraints on loading matrices.  Extensions handling covariates and alternative estimation procedures has been later discussed by  \citet{devito2023multi} and \citet{avalos2022heterogeneous}.  Similarly, \citet{roy2021perturbed}  proposed a perturbed factor analysis that focuses on inferring the shared structure while making use of subject-specific perturbations. Motivated by identification issues arising from these approaches, \citet{chandra2024sufa} recently proposed a class of subspace factor models which characterize variation across groups at the level of a lower-dimensional subspace where the study-specific factor loadings $\Gamma_s$ are obtained as $\Gamma_s = \Lambda \Delta_s$ where $\Delta_s$ is a  $d \times d_s$ matrix, with $d_s < d$, i.e. modeling the study-specific contributions as lower-dimensional perturbations of the shared factor loadings.

In all of these approaches there is however a modeling limitation, i.e. they allow precisely $S$ study-specific loading matrices $\Gamma_s$. In practice, this assumption is often violated. We provide some examples,  though a continuous spectrum of possibilities exists: (i) two or more studies exhibit high homogeneity and may share identical or nearly identical latent representations, and (ii) one study might contain highly heterogeneous subject groups, potentially forming distinct sub-populations with different latent structures. 
This could arise from minor experimental variations or unmeasured confounders. Additional  possibilities include (iii) the existence of a subset of units from different groups that share common characteristics, or (iv) instances where some units might simultaneously exhibit characteristics of multiple groups due to various reasons.

To address case (i), \citet{grabski2023bayesian} proposed a model allowing for partially-shared latent factors using a study-specific  {diagonal} selection matrix. The pattern of zeros and ones in this matrix determines which latent factors pertain to each specific study. The diagonal elements are then incorporated into a selection matrix governed by an Indian Buffet Process (IBP) prior \citep{ibp}. The same work discusses a variant, where study-specific labels are unknown and inferred via mixture models. While this approach moves towards addressing case (ii), it relies on the \textit{a priori} knowledge  of the total number of groups into which the units are divided and to our knowledge it has not been applied yet. 
Case (iii) reflects the fundamental limitation of single-partition approaches in characterizing complex factor patterns. Hence, despite these recent advancements, a unifying framework that seamlessly spans from classical factor models to MSFA, encompassing all the above mentioned cases or nuanced variants, remains absent. 

In this paper, we propose an alternative formulation for Bayesian MSFA leveraging the concept of informed or structured sparsity \citep{schiavon1, schiavon2, GriffinHoff}. This approach allows for flexible modeling of varying degrees of group heterogeneity in terms of latent factors contribution. Consistently with this, we name the proposed approach Adaptive Partition Factor Analysis (\model). \model$ $ employs shrinkage priors for latent factors, using the information contained in study labels, but without enforcing a deterministic structure. This structured shrinkage is combined with increasing shrinkage priors on the number of latent factors, following a cumulative shrinkage process approach \citep{legramanti,fruhwirth2023generalized}. This construction ensures identifiable separation of the shared and specific contributions while circumventing, under suitable regularity conditions, the well-known issue of rotational ambiguity in factor loadings for the specific factors, thereby allowing straightforward posterior interpretation of these quantities. Notably, we show that the proposed method can be seen as a particular neural network model allowing us to exploit further generalizations, including  the integration of subject- or study-specific continuous covariates in a flexible manner.

The next section details APAFA in terms of model specification,  prior distributions, and provides practical guidelines for prior elicitation.  We also explore the connections to neural networks and present generalizations. In addition, we address model identifiability. Section 3 reports  a comprehensive empirical assessment through simulation studies covering various data-generating scenarios. Section 4 reports two illustrative data analyses further highlighting the advantages of our approach in the motivating examples of bird species occurrence and genomics. The final discussion section explores generalizations and potential extensions. The code to reproduce our analyses is  available at  
\url{https://github.com/elenabortolato/APAFA}. 

\section{Adaptive partition factor analysis}

\subsection{Model specification}

Under the standard assumption that $\eta_i \sim N(0,I_d)$ and omitting the subscript $s$, we rewrite model \eqref{eq:msfa} as 
\begin{equation}
	y_{i} = \Lambda {\eta}_{i} + \Gamma {\varphi}_{i} + \epsilon_{i}, \quad  \epsilon_{i} \sim N(0, \Sigma),
	\label{eq:msfa2}
\end{equation}
where $\Gamma = (\Gamma_1, \dots, \Gamma_S)$ concatenates along columns all the study-specific factor loading matrices into a $p \times k$ matrix with $k = \sum_{s=1}^S k_s$ and ${\varphi}_{i}$ is a $k$-dimensional augmented vector containing the  ${\varphi}_{is}$ in Equation \eqref{eq:msfa} framed with suitable pattern of zeros. Specifically, all the latent factors pertaining to group $s$ will have $\sum_{l=1}^{s-1} k_l$ zeros, followed by subject-specific ${\varphi}_{is}$, followed by $\sum_{l=s+1}^{S} k_l$ zeros. {This formulation introduces two distinct factor loading matrices: $\Lambda$ containing loadings common to all units, and $\Gamma$ containing all the study-specific loadings that are activated following the non-zero patterns of $\varphi_i$.} An illustration is depicted in Figure \ref{fig:FA} \change{for $n = 15$, $p = 5$, $d = 2$, $k = 4$, $S = 3$, $k_1 = k_2 = 1$, and $k_3 = 2$.}

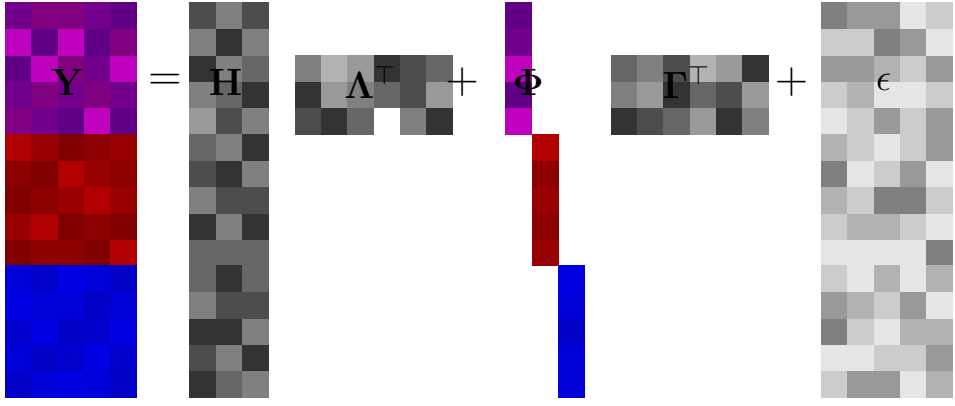
\begin{figure}
	\input{figure1}
	\caption{Multi-study Factor model representation: the $n\times p$ data matrix $Y$ (on the left) is written as the product of the latent factor  matrix H of dimension $n\times d$ by the factor loading matrix $\Lambda$ (the shared parts) plus the product of the latent factor  matrix  $\Phi$ of dimension $n\times k$ by the factor loading matrix $\Gamma$ (collecting  all the study-specific parts), and a random noise $\epsilon$. Different shades of purple, red, and blue identify the $S=3$ studies in $Y$. 
    }\label{fig:FA}
\end{figure}

As customary in factor analysis, one can marginalize out the {latent factors $\eta_i \sim N(0,I_d)$,} obtaining
\begin{equation}
	y_{i}  \sim N(\Gamma {\varphi}_{i}, \Lambda \Lambda^\top + \Sigma).
	\label{eq:msfamarg}
\end{equation}
Under this characterization, one can think to have, for each $i$,  a vector of dummy variables $x_i$ with $S$ entries characterizing the study to which unit $i$ belongs to and then assume that the $h$-th element of the vector $\varphi_{i}$ is  $\varphi_{ih} = \tilde{\varphi}_{ih} \psi_{ih}$ where $\tilde{\varphi}_{ih}$ is a continuous  random variable  and  $\psi_{ih} = f_h(x_i)$ with $f_h$ a deterministic activation function. In \citet{devito1,devito2}, for example, $f_h(x_i) = x_i^\top 1_S$ where $1_S$ is a $S$-dimensional vector of ones. 
%
In this paper, we propose to incorporate the information contained in the $x_i$'s in a more flexible manner, i.e. 
\begin{equation}
	\psi_{ih}= f_h( x_i^\top \beta_h), 
	\label{eq:nnet}
\end{equation}
keeping, for the moment,  $f_h$ as a general, non-linear function and considering the parameters $\beta_h$ as unknown. 

%

 It is now evident that Equations  \eqref{eq:msfamarg} and  \eqref{eq:nnet} transform model  \eqref{eq:msfa2} into a specific single layer neural network where the dummy variables in the $x_i$s are the input variables, $y_i$ are the $p$-dimensional output variables, $\varphi_i$s are the nodes of the hidden layer, $f_h$ are the activation functions, $\beta_h$ are the weights between the input and the hidden layer, and the elements in $\Gamma$ are the weights between the hidden layer and the continuous outcome $y_i$. The neural network representation of model \eqref{eq:msfa2} is reported  in Figure \ref{fig:NNET}.
    	In recent years, there has been a growing interest in Bayesian 
	analysis of neural networks. Typically, 
	the weight parameters are assigned a prior distribution and are learned updating the prior via Bayes rule. For comprehensive reviews, see \citet{goan2020bayesian}, \citet{fortuin2022priors}, and \citet{arbel2023primer}. Bayesian approaches to neural network models, however, not only apply the Bayesian inferential paradigm to learn the network parameters but also allow the neurons themselves to be stochastic \citep{neal1990learning,tang2013learning}. 
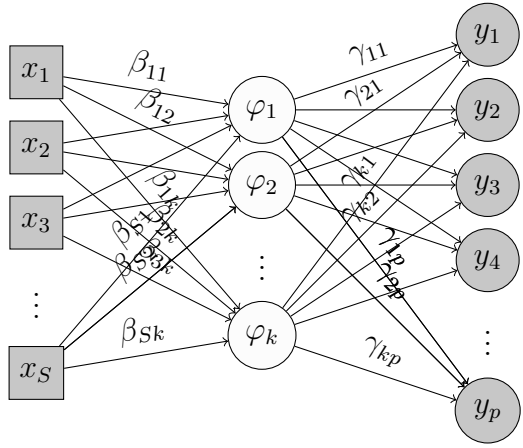
\begin{figure}
    \centering
    \input{neuron1}
    \caption{Neural Network representation: The input nodes are the categorical variables associated to the study structure. The first layer of latent variables are the latent study-specific factors.}
    \label{fig:NNET}
\end{figure}


 \subsection{Prior specification}

 We now define the prior structure for model \eqref{eq:msfa2}. Specifically, we let 
\begin{equation}  \label{eq:varphi}
	  \varphi_{ih}\sim N(0,\psi_{ih}(x_i)\tau_h^\varphi),
 \end{equation}  
where the variance of $\varphi_{ih}$ is a product of a global scale (depending on the index $h$) and a local scale (dependent on both indexes $i$ and $h$), and most importantly, on the variable $x_i$ through Equation \eqref{eq:nnet}.

We let the global scale parameter $\tau^\varphi_h \sim \text{Ber}(1 - \rho_h)$, where $\{\rho_h\} \sim \mbox{CUSP}(\alpha^\varphi)$  follows  a cumulative shrinkage process  \citep{legramanti}, i.e. for $h = 1, 2, \dots,$ $\rho_h= \sum_{l \leq h} w^\varphi_l$,  \: $w^\varphi_l = v^\varphi_l \prod_{m<l} (1 - v^\varphi_m)$, and $v^\varphi_l \sim \text{Beta}(1, \alpha^\varphi)$. 
This approach facilitates learning the number of study-specific factors or, following the neural network interpretation of our model, the number of neurons in the hidden layer. This construction, inspired by successful applications in factor models \citep{legramanti, schiavon1}, offers valuable insights into performing Bayesian inference for the number of nodes in the hidden layer of a Bayesian neural network, potentially connecting with existing literature on this topic \citep[e.g][]{wen2016learning, nalisnick2019dropout, fortuin2022priors, cui2022informative, jantre2023comprehensive, sell2023trace}.

We now move the discussion to the local scales $\psi_{ih}(x_i)$ that are assumed to be Bernoulli variables with probability depending on $x_i$. Specifically, we let
\begin{equation}  \label{eq:psix}
	\psi_{ih}(x_i)   \sim \text{Ber}\{ f_h(x_i^{\top}\beta_h)\},
\end{equation}
\change{with $f_h(\cdot) =  \text{logit}^{-1}(\cdot),$ for all $h$}.
Since $\varphi_{ih}(x_i) =\tilde  \varphi_{ih} \psi_{ih}(x_i)$, Equations \eqref{eq:varphi} and \eqref{eq:psix} imply
\begin{equation*}\label{eq:tilde_representation}
	  \tilde\varphi_{ih}(x_i) \sim   N(0,\tau_h^\varphi).
\end{equation*}
Marginalizing out $ \varphi_{ih}$ or $ \psi_{ih}$ yelds important insights. 
For instance, when marginalizing out $\psi_{ih}(x_i)$ while keeping $\tau_h^\varphi$ and $\beta_h$ fixed, we obtain:
\[
	  \varphi_{ih}(x_i) \sim \{1-\text{logit}^{-1}(x_i^{\top}\beta_h) \} \delta_0 + \text{logit}^{-1}(x_i^{\top}\beta_h) N(0,\tau_h^\varphi),
\]
where $\delta_a$ is a Dirac mass at value $a$.
{Conversely, marginalizing out $\tilde{\varphi}_{ih}$ while keeping $\varphi_i$ fixed and $\tau_h^\psi = 1$ leads to}
\begin{equation} \label{eq:marginalphi}	{y_{i}  \sim N(0, \Omega_i), \quad \Omega_i = \Lambda \Lambda^\top + \Gamma \text{diag}{\{\psi_i\}}\Gamma^\top+ \Sigma.}
\end{equation}
Notably, \eqref{eq:psix} corresponds to a simple form of a conditional variational autoencoder that, unlike the standard version, allows the latent covariance structure to be informed by the independent variables. Variational autoencoders are a specific extension of neural networks, widely recognized as descendants of classical factor models 
 \citep{kingma2014stochastic, sen2024bayesian}.
 
The benefits of the proposed solution should now be evident. The shrinkage prior on the elements  $\varphi_{ih}(x_i)$ enables and promotes, yet does not mandate, the sparse representation in \eqref{eq:msfa2}. Furthermore, it accommodates a wide range of scenarios discussed in the Introduction, including: (i) two or more studies exhibit high homogeneity and share nearly identical latent representations,  (ii) some studies involve highly heterogeneous groups of subjects, potentially resulting in two or more sub-populations with distinct latent structures, or any combinations of the above. The random partition of the observation induced by the sparse structure of the $\psi_{i}$ justifies the Adaptive Partition Factor Analysis (APAFA) name adopted for the proposed solution.

Another benefit of APAFA is the generalizability offered by its neural network interpretation. In fact, it is clear  that if other subject-specific covariates are available, say, $z_i \in \R^q$ we can include them as input variables, for example 
\begin{equation*}
	\psi_{ih}(x_i,z_i)   \sim \text{Ber}\{ \text{logit}^{-1}(x_i^{\top}\beta^{(x)}_h + z_i^\top \beta^{(z)}_h)\},
\end{equation*}
  thus allowing the latent factors $\varphi_i$ to have subject-specific, and not just study-specific, conditional distributions. For this reasons we will use henceforth the broader
  term \emph{specific factors (loadings)} to denote $\varphi$  ($\Gamma$) rather than the narrower term \emph{study-specific}.

We complete the prior specification for the specific part assuming standard prior for the factor loadings. Specifically for the general element $\gamma_{jh}$ of the factor loading matrix $\Gamma$, we let   $\gamma_{jh}\sim N(0,\zeta^\gamma_h)$, with  $\zeta^\gamma_h\sim \text{IGa}(a_\gamma,b_\gamma).$ 

One remark on the activation function pertains to the potential lack of properness of the conditional posterior distribution for the parameters $\beta_h$, $h=1, \dots,  k$.  This concern is  tied to the  issue of perfect separation in logistic regression. 
Assuming the group structure encoded by the $x_i$ dummy variables is both accurate and relevant to the problem at hand---in other words, that the original MSFA model by \citet{devito2} is correctly specified--- a perfect separation in the $\psi_{ih}$ may occur with respect to the $x_i$ within the logistic regression.
For this reason, the prior distribution on each parameter $\beta_h$ should be sufficiently concentrated around zero to prevent the associated conditional posterior from exhibiting a tendency toward monotonicity. 
To this end, we specify $\beta_h \sim N(0, B)$, where $B$ is a $S \times S$ diagonal matrix 
with elements proportional to $n^{-1}$. 
Thus, when perfect separation in the $\psi_{ih}$'s with respect to the $x_{i}$'s  occurs, the full conditional posterior is 
\begin{equation*}
	\mbox{pr}(\beta_{hs}|-)\propto \exp{\left\{-n_s \left[  \frac{ n}{2n_s} \beta_{hs}^2-\log \frac{e^{\beta_{hs}}}{1+e^{\beta_{hs}}}\right]\right\} }.
		\label{cond_post}
\end{equation*}
Writing $n=\bar{n}_s\times S$ where, $\bar{n}_s$ is the average number of units per group, we have
\begin{equation*}
	\mbox{pr}(\beta_{hs}|-)\propto \exp{\left\{-{n}_s \left[ \frac{ \bar{n}_s\times S}{2{n}_s} \beta_{hs}^2-\log \frac{e^{\beta_{hs}}}{1+e^{\beta_{hs}}}\right] \right\} } 
\end{equation*}
from which we can also recognize that   if the number of groups $S$  increases the posterior concentrates  close to 0, for any group size.  Despite the apparent inconsistency of a prior inducing  skepticism on the group configuration, it facilitates posterior inference. 

We now move to describe the prior specification for the shared part.  As before, we include an increasing shrinkage prior for the dimensions of the shared factors, i.e.
\[
\eta_{ih}\sim N(0,1), \quad \lambda_{jh} \sim N(0, \tau_h \zeta^\lambda_h),\quad
\zeta^\lambda_h\sim \text{IGa}(a_\lambda,b_\lambda),
\]
with a cumulative shrinkage prior on $\{\tau_h\} \sim \mbox{CUSP}(\alpha^\eta)$.  For the diagonal elements of $\Sigma$ we adopt an inverse gamma distribution and specifically $\sigma_j^2\sim \text{IGa}(a_\sigma,b_\sigma)$, which completes the prior specification. 

{Posterior inference is conducted via Gibbs sampling. See  the Supplementary Material.}

\subsection{Identification properties}

Identifiability of factor models is an overarching issue.
While these models are appealing in revealing possible interpretable relationships between a small set of latent factors and multivariate observations---even with complex dependency structures---this interpretability comes at a cost: factor analytic models {lack identifiability} in fundamental ways.  
For instance, loadings are not identifiable due to rotation, sign and permutation invariance. Rotational invariance allows transformation of loadings and factors by an orthogonal matrix, sign invariance arises when simultaneous sign changes leave the structure unchanged, and permutation invariance reflects the lack of intrinsic ordering in loading matrix columns. { See \cite{papastamoulis2022identifiability}  \cite{xu2023identifiable} \cite{fruhwirth2024sparse}  for a comprehensive review.}

Identification issues are exacerbated in the MSFA framework of \citet{devito2} where it may be difficult to disentangle the signal between shared and specific factors, a problem known as  \textit{information switching}. For instance,
\cite{chandra2024sufa} show how a linear combination of study-specific factors can instead be represented and reincorporated as part of the shared component. Building on these concepts, we provide conditions to
identify the shared from the specific factors in our setup. {Throughout, let $\orth_k$ denote the space of $k \times k$ orthogonal matrices.} We formalize the notion of information switching through the following definition:
\begin{definition} 
Let $n$ be the number of units, and $S_n$   the  number of distinct groups in model  \eqref{eq:marginalphi}, i.e.  $
S_n = \left| \cup_{i=1}^n \Omega_i \right|$. Denote with $\Psi$ the $n\times k$ matrix  that stacks in distinct rows all  $\psi_{i}=(\psi_{i1}, \ldots, \psi_{ik})$ and with $\Psi_h$ its generic column with, $\Psi_h\neq 1_{n}$ for all $h=1, \dots, k$.
{Let $\Omega_s^*$  and $\psi_s^*$ ($s=1, \dots, S_n$) be the distinct values of $ \Omega_i$ and $\psi_i$, respectively. Similarly, let $W_s^*=\Omega_s^*-\Lambda\Lambda^\top-\Sigma= \Gamma  \text{diag}\{\psi_s^*\} \Gamma^\top$.} The model suffers from information switching if there exist   $\tilde \Gamma\ne \Gamma$ and  $ \tilde \Psi \ne \Psi$ such that
$W_s^*= \tilde\Gamma {\mbox{diag}\{\tilde\psi_s^*\} }\tilde \Gamma$ for all $s$,  with  $\tilde \Psi_h=1_n$ for at least one $h$. \label{def:infosw}
\end{definition}
It is clear that if $\tilde \Psi_h=1_n$, the  $h$-th factor could no longer be interpreted as specific and should be moved to the shared part. Consistently, $\Lambda\Lambda^\top$ would represent only a fraction of the shared variance.
Note that in Definition \ref{def:infosw},  $S_n$ does not refer to the number of distinct labels or studies $S$ fixed \textit{a priori} and embedded in the categorical variables $x_i$, but rather it is equal to the number of distinct configurations of the rows of the matrix $\Psi$. The following theorem ensures the identification in terms of information switching.


 \begin{theorem}\label{thm:identifiability}
For the model defined in  \eqref{eq:marginalphi}, {if $\Psi_h\ne 1_n$ for all $h \in  \{1,\ldots, k\}$} and 
 $\Gamma$ is of full column rank $k$ with $k<p(p+1)/2$,   then the model is resistant to information switching. 
 \end{theorem}
\begin{proof}
 Let $\Psi^*$ be the $S_n \times k$ matrix containing the distinct $\psi^*_s$. Since each $\psi_s^*$ is a vector of zeroes and ones, we can see that  
$W_s^*=   \Gamma   \text{diag}\{\psi_s^*\}  \Gamma^\top=    \Gamma   \text{diag}\{\psi_s^*\} \text{diag}\{\psi_s^*\}\Gamma^\top.$ 
 As in any factor analytic decompositions, if $P_s\in \orth_k$ we can write
$W^*_s =  \Gamma \text{diag}\{\psi_s^*\} P_s P_s^\top \text{diag}\{\psi_s^*\}\Gamma^\top.$
Consistently, each element $w^{*(jl)}_s$ of $W_s$, { for $j,l=1, \ldots, p$} can be written as
\begin{equation*}
w^{*(jl)}_s=\sum_{h=1}^k \gamma_{jh}\gamma_{lh} {\bar\psi_{hs}\bar\psi_{hs}^\top},
\end{equation*}
where $\bar \psi_{hs}$ is the $h$-th column of the matrix product $\text{diag}\{\psi_s^*\} P_s.$ Since $P_s\in \orth_k$, the product $\bar \psi_{hs}\bar\psi_{hs}^\top$ is either 0 or 1 by construction. Thus, we know that there exist binary variables defined as $\tilde \psi^*_{hs}= \bar \psi_{hs}\bar\psi_{hs}^\top$ for $h=1, \dots k$  leading to the same $W_s^*$ for each $s=1, \dots, S_n$. Hence, it is sufficient to check whether these $\tilde \psi_s^*$ are actually  different from the original $\psi_s^*$ for each $s$. 
{To this end, for the pair $(jl)$ consider 
\[
\sum_{h=1}^k \gamma_{jh}\gamma_{hl} \tilde \psi_{hs}^* = w^{*(jl)}_s
\]
as a linear equation in  $\tilde \psi_{hs}^*$. Since we have $S_n p(p+1)/2$ linear equations of this type, we can frame them into a linear system $G\mathbf{x} = \mathbf{w}$ defying $G$ as the $S_n p(p+1)/2 \times k$ matrix of known coefficients, $\mathbf{x}$ as the vector of variables, here associated to the $\tilde \psi_{hs}^*$, and  $\mathbf{w}$ as the vector stacking 
the half-vectorization of \(W_s^*\) (containing its upper triangular entries) for $s=1, \dots, S_n$. 
Specifically,  the full design matrix \(G = I_{S_n} \otimes G_s \), where $G_s$ is the $ p(p+1)/2 \times k$ matrix whose rows correspond to index pairs \((j, l)\) with \(1 \leq j \leq l \leq p\), and whose entries are given by
\[
G_{(j,l) h} = \gamma_{jh} \gamma_{lh}, \quad \text{for } h = 1, \dots, k.
\]
Since  \(\text{rank}(G_s) = k\), then \(\text{rank}(G) = S_n k\), and  $\text{rank}(G)=\text{rank} (G|\mathbf{w})$ by construction. Thus, since \(k < p(p+1)/2\) by assumption, the solution is unique by the Rouché–Capelli Theorem. The uniqueness of the solution for the elements $\psi_s^*$ ensures that it is impossible to find a representation where for one $h$, $\tilde{\Psi}_h$ is a vector of one. }
\end{proof}

{
The proof of Theorem \ref{thm:identifiability} is based on proving the existence of unique vectors $\psi_s^*$ for the specific factors. This guarantees that no alternative representation exists in which, for some $h$, the corresponding $\tilde{\Psi}_h$ is a vector of ones. Importantly, it also establishes the uniqueness of the sparsity patterns encoded by the $\Psi$ matrix. That is, any alternative representation inducing the same covariance structure must preserve the subject-specific activation patterns of the specific components. The following corollary formalizes this consequence.
}
\begin{corollary}
Under the conditions of Theorem \ref{thm:identifiability}, the sparsity indicators $\psi^*_{hs}$ are uniquely identified up to a permutation of the specific factors. 
\end{corollary}

Based on the identifiability findings of Theorem \ref{thm:identifiability}, the practical conditions that ensure resistance to information switching in APAFA are discussed in  Lemmas \ref{lemma:priorGammaTrunc} and \ref{lemma:priorGamma}.

\begin{lemma}\label{lemma:priorGammaTrunc}
   Under any continuous prior for the elements of $\Gamma$ and truncating its number of columns to $K=p(p+1)/2-1$, the information switching for model \eqref{eq:marginalphi} has zero probability.
   \end{lemma}
   
\begin{lemma}\label{lemma:priorGamma}
   Under any continuous prior for the elements of $\Gamma$, for any $\alpha^\varphi < \varepsilon p^{2}/2 $ and sufficiently small $\varepsilon>0$, the information switching for model \eqref{eq:marginalphi} has prior probability bounded from above by $\epsilon$.
   \end{lemma}
\begin{proof}
From Markov's inequality  $\text{pr}(k>2(p+1)/2)\leq\frac{ \alpha^\varphi}{p(p+1)/2} \leq \varepsilon$.
\end{proof}

A key advantage of our model is that the requirement of $k<p(p+1)/2$ is minimal, as it is customary in factor analysis to have latent factors of a considerably  lower dimension than the observed data ($k<p$). Notably, \cite{chandra2024sufa}, in different settings, necessitate the dimension of the specific part to be smaller than that of the shared part implicitly considering the study-specific part as a noise or disturbance parameter. Although we are in different settings, where the number and composition of the groups is not pre-specified, we do not restrict $k$ to be also smaller than $d$. 
This is crucial in many contexts. For instance, in cancer studies, it is possible to study very different types of cancer simultaneously, identifying a range of specific factors that describe the specific characteristics of each type of cancer, while also uncovering a few common factors that are shared across cancers, regardless of the organ or tissue affected, possibly pertaining to fundamental mechanisms of oncogenesis.
 %


 {
Under stronger assumptions, detailed below, we show that the specific factors are not only distinguishable from the shared ones, but also fully identifiable up to a permutation of their columns.
We start providing the following definition.  }
\begin{definition}[Non-replicable Sparsity Pattern Condition]
 \label{def:nrspc}
{All columns of $\Psi^*$, where $\Psi^*$ is the $S_n\times k$ matrix stacking all the distinct $\psi^*_s$,  are different.}
\end{definition}

{
We now provide a result that, under  the non-replicable sparsity pattern condition, ensures full identifiability of the specific factor, up to permutation. }

 \begin{theorem}\label{thm:full_ide}
Let \( \Gamma \in \mathbb{R}^{p \times k} \) be a real matrix with full column rank $k.$ 
If $P\in \orth_k$  and $\Gamma' = \Gamma P$ is a rotation of the specific factors, under the Non-replicable Sparsity Pattern Condition (Definition~\ref{def:nrspc}), then
$$
\Lambda \Lambda^\top+ \Gamma \Gamma^\top +\Sigma= \Lambda \Lambda^\top+ \Gamma' \Gamma'^\top +\Sigma, $$
and, for each $s=1, \dots, S_n$
$$
\Lambda \Lambda^\top+ \Gamma \Psi_s^* \Gamma^\top +\Sigma = \Lambda \Lambda^\top+ \Gamma' \mbox{diag}(\psi_s^*) \Gamma'^\top +\Sigma, $$
if and only if $P$ is a permutation matrix.
\end{theorem}

\begin{proof}
We begin by fixing the ordering of the columns of $\Gamma$. Under this fixed ordering, we will show that the only admissible orthogonal matrices $P \in \orth_k$ are diagonal with entries $\pm 1$. Moreover, applying any permutation to the columns of $\Gamma$ initially and repeating the argument would yield the same conclusion, i.e. the only admissible $P$ relative to the permuted ordering are diagonal with entries $\pm1$. Consequently, the general solution requires $P$ to be a signed permutation matrix. 

Let \( A_s  \) be the active support of group \( s \), namely $$A_s := \{ h \in \{1, \dots, k\} \,:\, \psi^*_{hs} = 1 \},$$ and define \( \Gamma_s \in \mathbb{R}^{p \times k_s} \) to be the submatrix of \( \Gamma \) consisting of the columns indexed by \( A_s \). Notably, these are identifiable thanks to Corollary 1.

The proof proceeds considering that  \( P_s \in \orth_{k_s} \), leaves $ \Gamma'_s\Gamma'_s{}^{\top} $ invariant and hence that the marginal distribution of group $s$ is also invariant under such rotations.
Given these \( P_s \in \orth_{k_s} \) for each \( s = 1, \dots, S_n \), we investigate whether a ``global'' orthogonal matrix \( P \in \mathcal{O}_k \) can be constructed by assembling the blocks \(P_s\). This requires consistency across overlapping blocks, ensuring that each \( P_s \) remains orthogonal, for example satisfying $\sum_{l \in A_s}  \|P_{s,hl  }\|=1$ for all $h$ and $s$.
We will argue in the following that, under the Non-replicable Sparsity Pattern Condition of Definition~\ref{def:nrspc},  this is possible only if  $P$ is a diagonal matrix with entries $\pm 1$.
For any  $h=1, \ldots, k$, there exist at least one group loading factor $h$, i.e. $ \psi^*_{sh}= 1$ for at least one $s$.
Then we have two possibilities: 
 $i)$ $\psi_{s'h}^*=0$ for all other $s'\neq s$, or 
 $ii)$ $\psi_{s'h}^*=1$ for at least another $s' \neq s$.
In the first case $i)$, $h$ is a factor belonging exclusively to  group $s$. Clearly, by Definition~\ref{def:nrspc} it is the only exclusive factor. If this is the only factor of group $s$, then the rotation matrix $P_s$
reduces to a scalar. Since  $P$ is orthogonal, $\|P_{h }\|^2=1$ implies $P_{hh}=\pm 1$
and all other entries in the $h$th row/column of $P$ must vanish.
If $h$ is not the only factor of group $s$, then $A_s$ contains at least another $h'$ that is shared among $s$ and at least another $s'$. In this case,  $h'$ is in case $ii)$.
Consider then  case $ii)$. In order to be orthogonal, matrix $P$ needs to satisfy
\begin{equation}
\label{eq:conditionA} \:  1=\|P_{h  }\|^2=\sum_{l =1}^k P_{lh}^2 = \sum_{l  \in A_{s'} \cap A_s} P_{lh}^2+
\sum_{l  \in A_{s'} \setminus A_s} P_{lh}^2+ 
\sum_{l \in A_s \setminus A_{s'}}P_{lh}^2+
\sum_{l \notin A_s \cup A_{s'}}P_{lh}^2 .  
\end{equation}
Additionally, since $h\in A_s$, then
\begin{equation}
\label{eq:conditionB}   \: 1=\|P_{ {s},h  }\|^2
=\sum_{l  \in A_{s} \cap A_{s'}} P_{lh}^2+
\sum_{l  \in A_s \setminus A_{s'}} P_{lh}^2
\end{equation}
and similarly, since $h\in A_{s'}$,
\begin{equation}
\label{eq:conditionC}   \: 1=\|P_{ {s'},h  }\|^2
=\sum_{l  \in A_{s'} \cap A_{s}} P_{lh}^2+
\sum_{l  \in A_{s'} \setminus A_{s}} P_{lh}^2.\end{equation}
Equations  \eqref{eq:conditionA}, \eqref{eq:conditionB}, and \eqref{eq:conditionC} enforce that all $P$, $P_s$, and $P_{s'}$ are orthogonal matrices. These set of equations take the form: 
$a+b+c+d=1,\: a+b=1,\:a+c=1$ with $a,b,c,d\geq0$. This system admits only the solution $a=1$ and $b=c=d=0$. 
This observation directly constrains the entries of the $h$-th column of \( P \) pertaining to the last three summations of \eqref{eq:conditionA} to be zero.  Consequently, the only actual rotations, i.e. those induced by non-diagonal $P_s$, are those corresponding to factors \emph{shared} between groups \( s \) and \( s' \), i.e. the $h$-th column of $P$ must vanish outside $A_{s} \cap A_{s'}$.

Assumption \ref{def:nrspc} induces  further constraints on \(P\). Specifically, for any pair \((h, h')\) jointly in groups $s$ and $s'$, there exists a third group \(s''\) such that either 
\(h \in A_{s''}\) but \(h' \notin A_{s''}\), or  
\(h' \in A_{s''}\) but \(h \notin A_{s''}\).  
Applying the same norm argument to $s$ and $s''$ (or equivalently to  $s'$ and $s''$)  forces the \(h\)th column of \(P\) to vanish outside \(A_s \cap A_{s''}\). Crucially, since \(A_s \cap A_{s''}\) excludes \(h'\) (or \(h\)), we conclude $P_{hh'}=P_{h'h}=0$ and $P_{hh}=P_{h'h'}=\pm 1$.  The argument is valid for all $h$, leading to $P$ being a a diagonal matrix with entries $\pm1$.  
\end{proof}

Clearly,  the shared variability induced by $\Lambda$ 
remains susceptible to well-known identifiability issues. These challenges can be effectively addressed using state-of-the-art  methods \change{for post-processing the MCMC output} such as those proposed by \citet{poworoznek2021efficiently} and \citet{papastamoulis2022identifiability}  \change{or by imposing additional structural constraints on the loading matrix, as in \citet{fruhwirth2023counts, fruhwirth2024sparse}.} 
 \change{Interestingly, the constraints introduced in the latter works are formulated in terms of sparsity in the loading matrix. Although our identifiability conditions are imposed on the factors, whereas \citet{fruhwirth2023counts} focus on the factor loadings, both approaches exploit the concept of sparsity, highlighting its central role in achieving identifiability in factor models.}

\section{Simulation study}\label{sect:sim}

\begin{figure}[t]
\label{fig:simscenarios}
    \centering
    \input{scenarios}
    \vspace{-1cm}
    \caption{Scenarios' true sparsity pattern (first four plots) and posterior estimate for a generic replicate (last four plots). 
    }
    \label{fig:scenarios}
\end{figure}
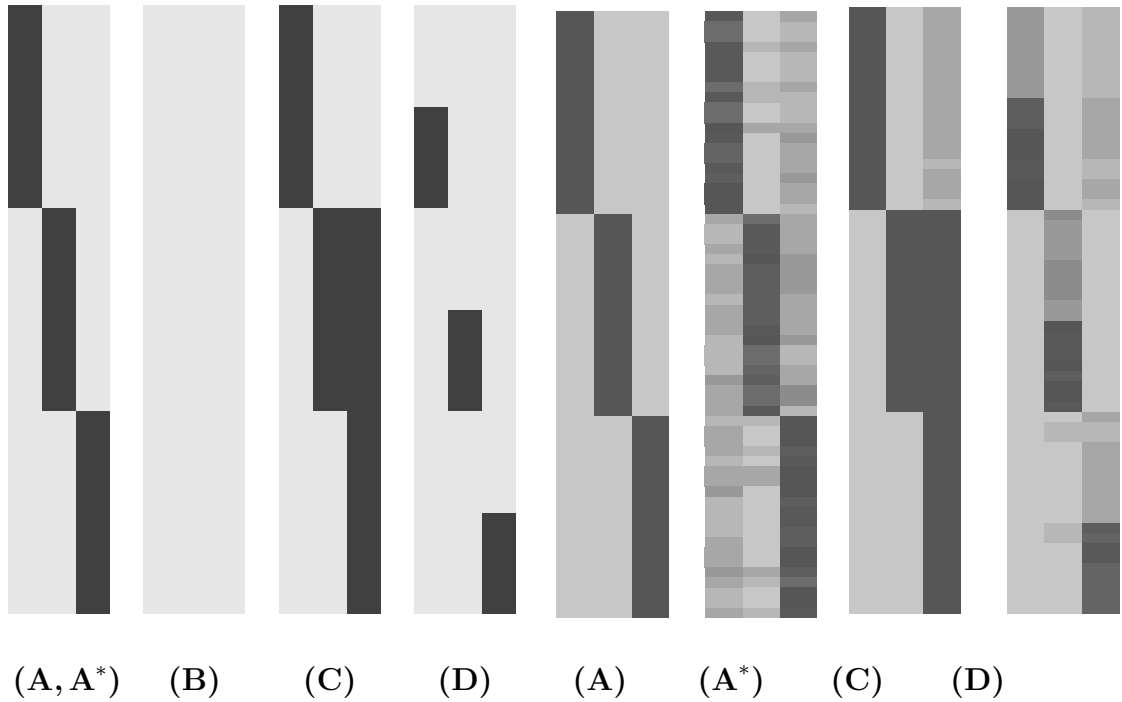

We assess the performance of APAFA in a simulation experiment and its relative merits with respect to the approach of \citet{grabski2023bayesian}  (TETRIS).  We consider two main configurations concerning the data dimensionality: tall data where $n=60$ and $p=10$ and large data where $n=45$ and $p=60$  and for each of them, simulate data from different scenarios. 
Scenario A represents the case in which the original MSFA model is correctly specified. Namely, we have $S=3$ groups with sample sizes $n_1=n_2=n_3=n/3$ and each of them has a single different specific factor. The total number of shared factor is $k=3$. The groups labels match exactly the three groups. To test our model's ability to detect the underlying structure, we also fit the model without providing the labels—meaning the model is unaware of the existence of three distinct groups. We refer to this situation as Scenario $\text{A}^*$. 
In Scenario B there are $S=3$ studies that are perfectly homogeneous, i.e. $k_s=0$ for each $s=\{1,2,3\}$ or equivalently $\Phi$ has all entries equal to zero. Notably, we fit each competing model providing the information that $S=3$ groups are present. 
Scenario C presents a complex situation in which there exist $S=3$ studies with partially shared latent factors. Specifically, study 1 and 2 present their own factors ($\varphi_1, \varphi_2)$ as in Scenario A, while a third factor is present for units of the second and the third study. 
 Scenario D addresses a situation in which for each of the three studies a subset of units has no specific factors. 
%
The left panels of Figure \ref{fig:simscenarios} provide a representation of the true specific latent factors sparsity pattern. Further details on the data-generating process are reported in the Supplementary Material.

Each combination of scenarios and dimensions is replicated independently $R=10$ times. For each of them, all the competing methods have been fit. {For the scenario where $n>p$, the computation time for a single replicate ranged from 3 to 4 minutes, whereas for the case $p>n$, it spanned from 10 to 13 minutes on a laptop CPU with a clock speed of
1.00 GHz. More details on the computing time are available in the Supplementary Material.}
Posterior estimates for all methods are obtained by running MCMC algorithms for each of the $R$ replicates for 10,000 iterations,  discarding the first 8,000, and then keeping the last 2,000 values. An implementation of  the Gibbs sampler  {reported in the Supplementary Material}  is available at
\url{https://github.com/elenabortolato/APAFA}.
For TETRIS we used the implementation available at \url{https://github.com/igrabski/tetris}.

To quantify the performance, we first focus on quantities of interest that are invariant from factor rotations. Specifically, we  consider the posterior mean of the number of active factors $d$ and $k$, along with measures assessing the accuracy of estimating the total variance matrix  $\Omega_s= \Lambda \Lambda^\top+ \Gamma_s\Gamma_s^\top + \Sigma$ and partial variance $\Lambda \Lambda^\top$.  We first compute the posterior means of these quantities and subsequently compare them to the corresponding true matrices using RV coefficients  \citep{abdi2007rv}, defined as 
$$
RV_{E,T} =\frac{\Tr(E^\top T)}{\sqrt{\Tr(E^\top E)\Tr(T^\top T)}},
$$
where $E$ and $T$ are both symmetric positive definite  matrices representing the estimated and true matrix, respectively. The RV coefficient takes values in (0,1) with values close to 0 for very dissimilar matrices, and values close to 1 for highly similar matrices. 
\input{tables.tex}
Tables \ref{table:scen_n>p} and \ref{table:scen_n<p} report the results. Table \ref{table:scen_n>p} shows that both methods are comparable in terms of posterior concentration around the true number of factors with a slightly better performance by APAFA. \change{Scenario A is the most favorable for TETRIS, as the model is correctly specified. Indeed,  although TETRIS struggles to disentangle the specific and shared components, it still provides overall good estimates of the covariance matrices with slightly better or comparable performance to APAFA. Considering  the standard errors, such superiority is not striking.}   TETRIS has not been fitted in scenario A$^*$, since no labels are assigned to units,  and in scenario D, since the method is not suited to  find subgroups of units with different factors within a given group.

Qualitatively similar results are observed in the $n<p$ case, reported in Table \ref{table:scen_n<p}.  
The performance of estimation of the shared variance, often the interest of the analysis, is reported in Figure \ref{fig:LL2}. 
In Scenario B, characterized by the absence of specific factors, TETRIS successfully recognizes the shared structure. However, it faces challenges in distinguishing the contributions to the variance components when group-specific characteristics are present. 
\begin{figure}[h!]
    \centering
        \includegraphics[width=.87\linewidth]{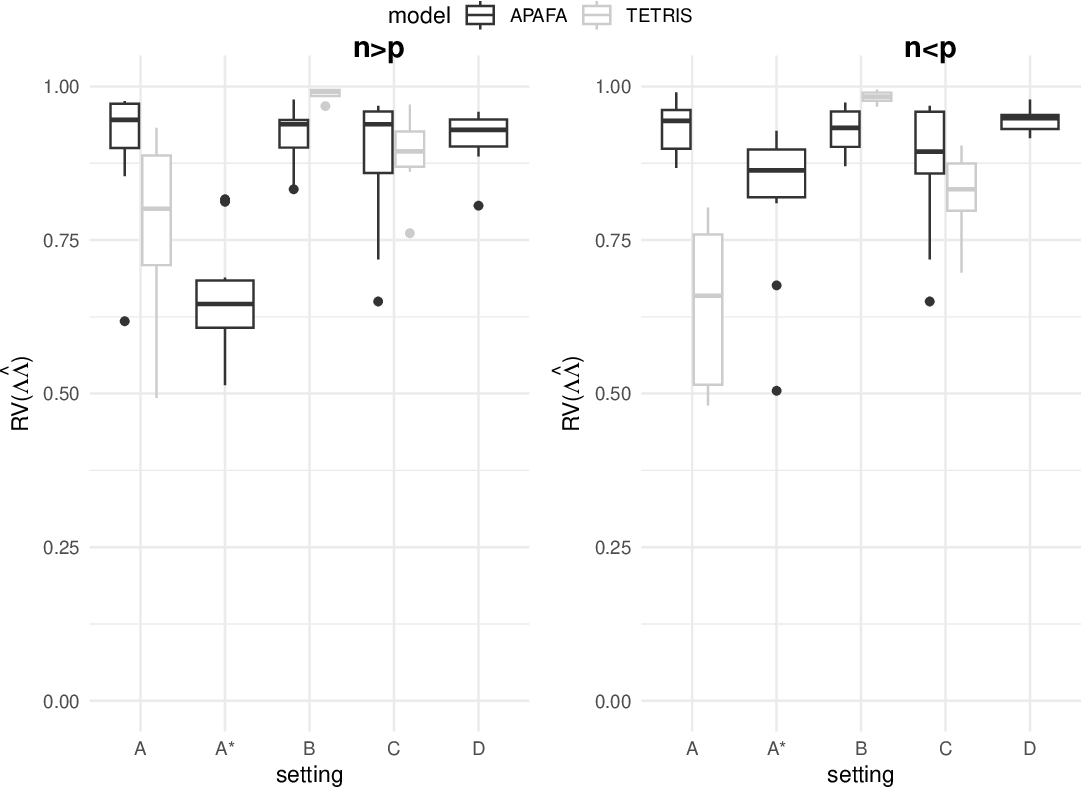}
    \caption{Monte Carlo distribution of the RV coefficient for the shared variance component under configuration $n<p$ (left panel) and $n>p$ (right panel).} \vspace{-0.1cm}
    \label{fig:LL2}
\end{figure}
\begin{figure} 
    \centering
    \includegraphics[width=0.60\linewidth]{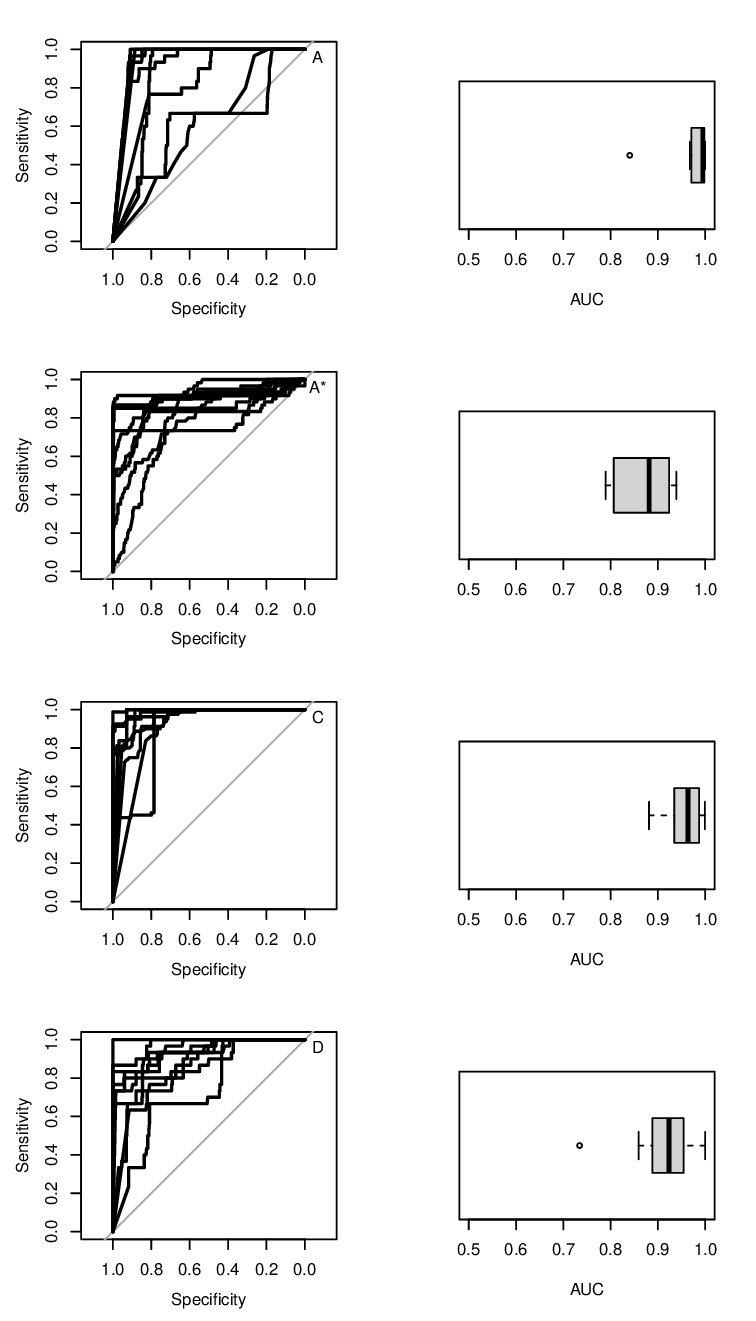}
  \caption{ROC curves (\textit{left}) and distribution of the AUC (\textit{right})  \change{obtained from posterior probabilities under the APAFA model over 10 independent replicated datasets, illustrating the accuracy of factor-to-unit assignments across scenarios} (from top to bottom: A, A$^*$, C, D), with $n>p$.}
    \label{fig:ROC}
\end{figure}
To assess the sparsity of the specific factors and the induced adaptive partitioning of units into subgroups that do not deterministically align with the grouping variable $x$, we compared the true zero/one pattern used to simulate the synthetic data reported in Figure~\ref{fig:scenarios} with the estimated $\psi_{ih}(x_i)$. To address the issue of {column permutations}, for the sake of this simulation evaluation, when the true value of $\Phi$ is known, we post-processed the order of the estimates of $\Phi$. Details are reported in the Appendix.
Figure \ref{fig:ROC} displays the ROC curves along with the corresponding area under the curve (AUC) for each replicated dataset. The results indicate a strong ability of APAFA to detect the true partition induced by the sparsity pattern of the specific factors in $\Phi$ in each scenario.

Additional assessments, including a sensitivity analysis and a simulation with nearly identical study-specific factors, further confirm APAFA’s empirical performance (see Supplementary Materials for details).

\section{Real data illustrations}

In this section, we analyze two datasets within the motivating contexts of animal co-occurrence studies and genomics. Rather than replicating previous analyses, we emphasize specific aspects allowed by the methodological advancements of the proposed model. In particular, we focus on qualitative insights that APAFA is able to reveal, offering valuable perspectives that are beyond the reach of existing state-of-the-art methods.

\subsection{Bird species occurence dataset}
\label{sec:birds}
We first examine   the co-occurrence patterns of $p=50$  bird species in Finland analyzing a data set collected over nine years (2006-2014) across $S=200$ locations \citep{lindstrom2015large}. 
The average number of sightings per location during the entire period is about 5, for a total of $n=914$ sites examined over the years.  
Shared factors may depend on ambient characteristics such as temperature, latitude, habitat type, and proximity to the ocean. To illustrate the performance of the proposed model in this context, we avoid fine model specification including these covariates, and \change{rather examine  if } the estimated latent factors are able to reconstruct some of this information.

We model species presence or absence using the multivariate probit regression model, 
\[
y_{ij} =\mathds{1}(z_{ij} >0),\quad  z_{i}  = \Lambda {\eta}_{i} + \Gamma {\varphi}_{i} + \epsilon_{i}, \quad  \epsilon_{i} \sim N(0, \Sigma).
\]
 We set the hyperparameters governing \textit{a priori} the number of active common and specific factors to $\alpha^\eta=15$ and $\alpha^\varphi=15$ respectively. We run the Gibbs sampler for 10,000 iterations and use the last 5,000 to obtain posterior estimates.
 Five active factors were identified in both shared and specific components.
Among the various aspects, we focus on interpreting the main specific factors, $\varphi_{i1}$.
 This factor takes non-zero values at a subset of locations, effectively grouping observations across multiple sites. This behavior aligns with the possibility of a partially-shared latent factor, similarly to \citet{grabski2023bayesian}. A \emph{post-hoc} analysis of these locations, using a categorical covariate that describes habitat type (not used in the model fitting), revealed that these locations are predominantly urban. Consistently with this, it is clear that the proposed method is able to find subject-specific factors associated with unobserved grouping variables.  See the first panel of Figure \ref{fig:birds_habitat} where the specific factor matrix $\Phi$ is reported with the columns ordered by habitat type. Notably, this finding aligns with recent ornithological literature examining the impact of urbanization on bird interaction \citep{pena2023relationships}. 
Similar qualitative insights can be appreciated in the second panel of the same figure, where the columns are ordered by the associated latitude values.  Interestingly, both the first and fourth  factors are strongly associated with latitude. Uncertainty quantification for these figures is provided in the Supplementary Material.
\begin{figure}[h!]
    \centering \vspace{-0.28cm}
    \includegraphics[width=0.9\linewidth]{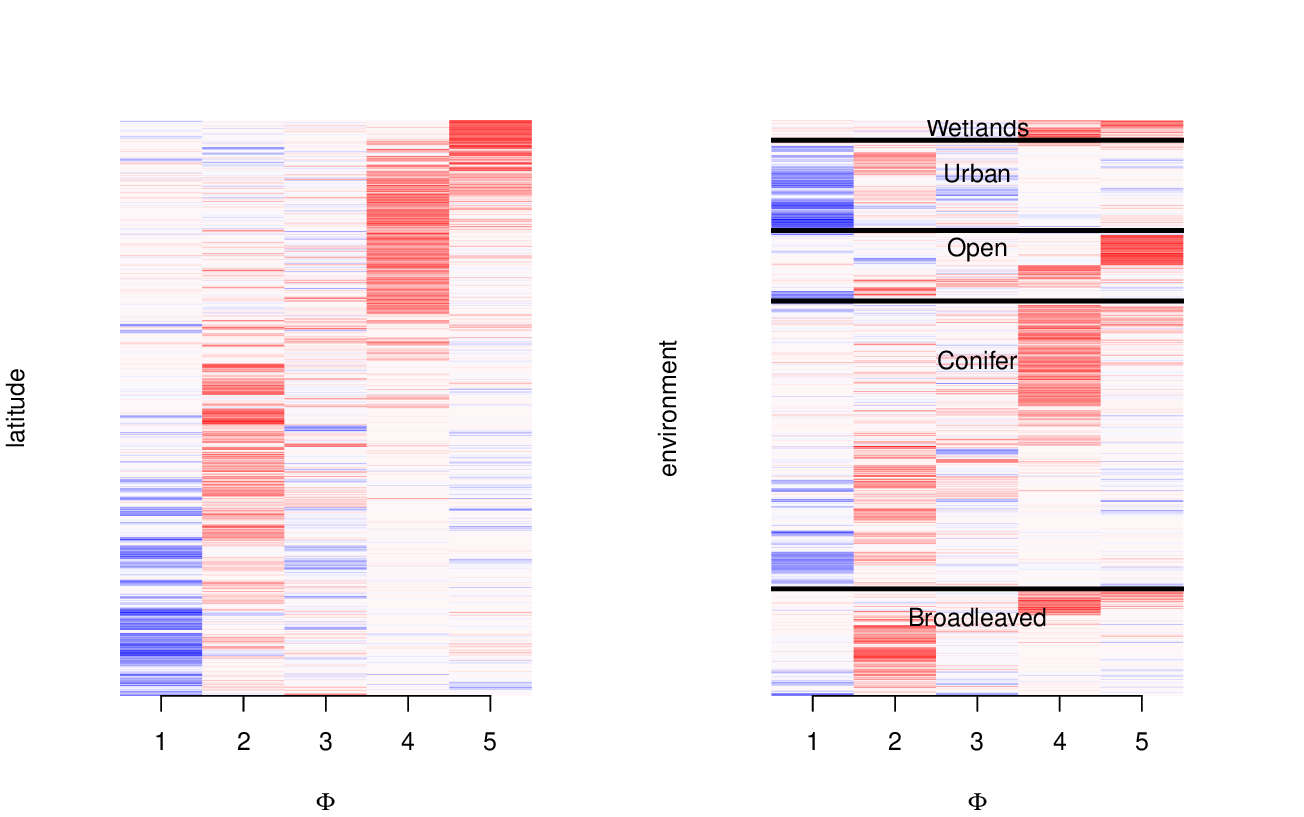}
    \caption{Posterior estimate of specific factors in the bird data example ordered by latitude (\textit{left}) and by habitat type (\textit{right}). Colors range from blue for negative values to red for positive values.}
   \label{fig:birds_habitat}
\end{figure}
\subsection{Immune dataset} \label{sec:immune}
As a second example, we analyze a publicly accessible dataset that contains transcriptomic data for $p=63$ genes linked to immune system function in female oncology patients with ovarian cancer. 
This dataset can be accessed through the \texttt{curatedOvarianData} package in Bioconductor \cite{ganzfried2013curatedovariandata}. The same dataset was also utilized in the seminal MSFA paper by \citet{devito1}. This dataset comprises four studies: GSE9891 and GSE20565, which utilize the same microarray platform for data acquisition, as well as TCGA and GSE26712, with respective sample sizes of 285, 140, 578, and 195 units ($n=1198$ observations in total). We fit model \eqref{eq:msfa2}  with hyperparameters governing \textit{a priori} the number of active common and specific factors to $\alpha^\eta=1$ and $\alpha^\varphi=4$,  respectively. The chain was initialized from $d=k=12$ active factors. 
The shape and rate parameters chosen for the Inverse Gamma prior for the elements of $\Lambda, \Gamma$ and the diagonal elements of $\Sigma$ were all equal to 2. We run the Gibbs sampler for 10,000 iterations and use the last 5,000 to obtain posterior estimates.

Posterior analysis identified 5 shared factors and 8 specific factors. Notably, the number of shared factors is consistent with the findings of \cite{devito1} using the BIC criterion.  As in the previous section, we compared the resulting matrix of activated specific factors to covariate data not included in the initial analysis, which are also available in the \texttt{curatedOvarianData} package. This comparison provides additional validation and insight into the biological relevance of the factors identified.  Figure  \ref{fig:phi}  shows that the grouping associated with the studies is captured within the structure of $\Phi$. Specifically,  factors $\varphi_6$ and $\varphi_7$  are clearly related to the group structure provided by  $x$, in particular with the TCGA and Gaus experiments, while individuals from the GSE9891 and GSE20565 studies, both using the same microarray platform for data processing, do not present specific variance adjustments.
The remaining factors are linked to a small number of units each, resulting in numerous configurations that far exceed the initially presumed four. Upon examining external covariates not included in the analysis, we found that $\varphi_1$ is associated with 6 units that are distinct from the rest, as they exhibit the clear cell histological subtype. Factor 3 distinguishes the gene expression of 7 patients with mucinous carcinoma  from other types of carcinoma and 8 patients with tumors in the fallopian tubes {(\texttt{ft})} rather than in the ovary. 
 Factors 4 and 8 discriminate patients with tumors classified as \texttt{other}  (sarcomatoid, adenocarcinoma, dysgerminoma), along with some  units presenting clear cell and endometrial tumors (\texttt{endo}, 26 units). See the boxplots in Figure   \ref{fig:phi} for details.
\begin{figure}[h!]
    \centering
   \includegraphics[width=0.96\linewidth]{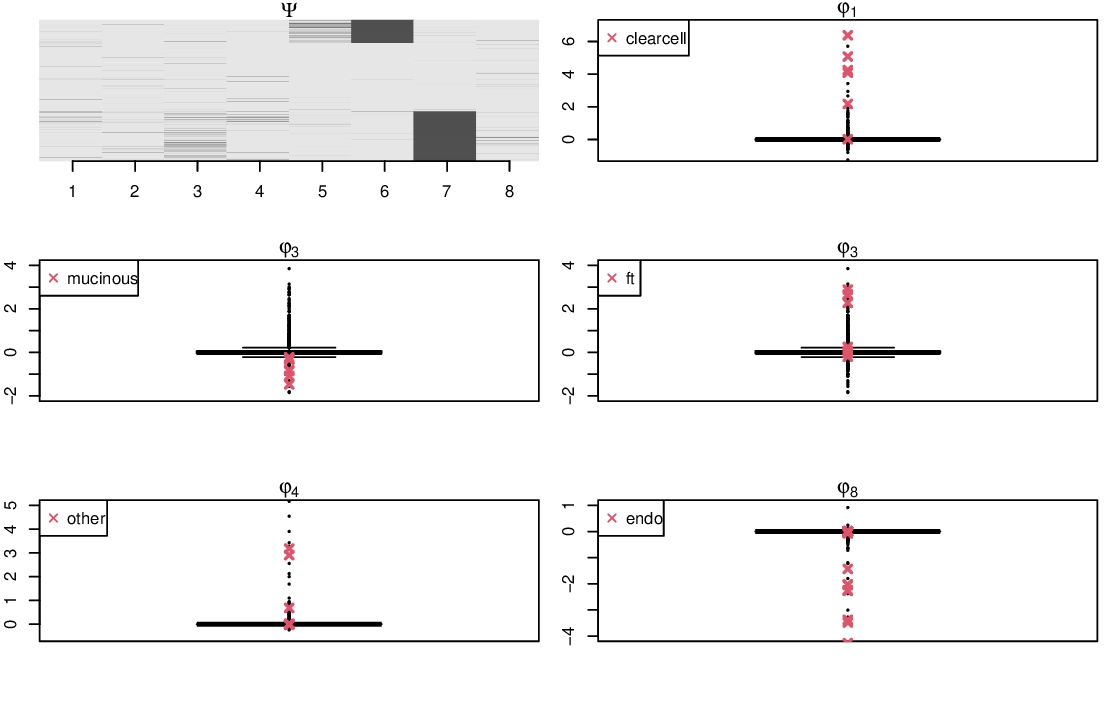}
    \caption{Posterior means of each local scale $\psi_{ih}$ associated to specific factors  (\textit{first panel}) and distribution of specific factors (boxplots) compared to external covariates (indicated with a cross). }
    \label{fig:phi} 
\end{figure}
\begin{figure}[h!]
    \centering
      \includegraphics[width=0.96\linewidth]{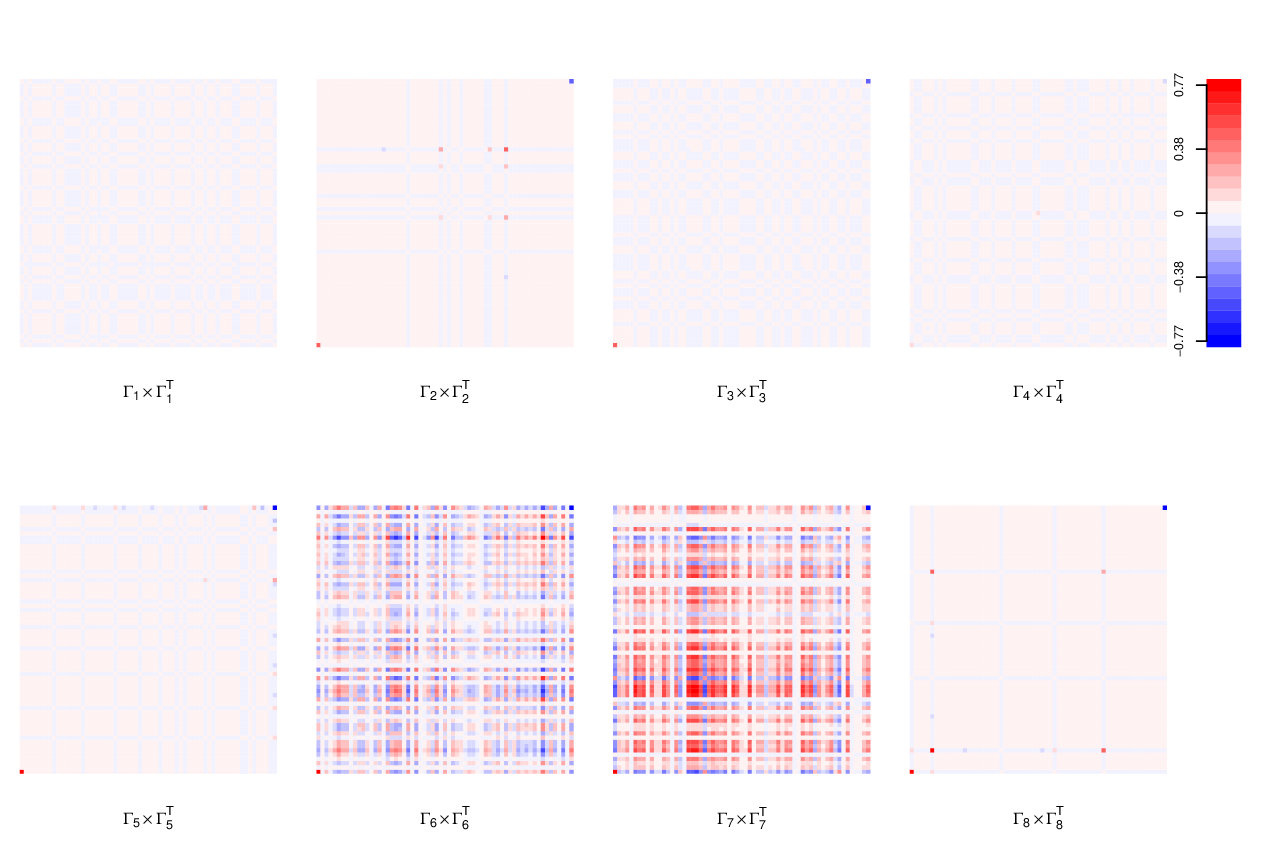}
    \caption{Estimated contribution to covariance matrix of specific factors in the gene expression case study. } 
    \label{fig:GammaGammaimmune}
\end{figure}

Figure \ref{fig:GammaGammaimmune} illustrates the contributions to the correlation matrix given by $\Gamma_h\Gamma_h^\top$ for $h=1,\ldots,8$.
The comparison with Tetris, eliciting a prior coherent with that used in APAFA, i.e., 5 factors in total, with 4 specific, revealed instead  46 active factors, with a 90\% credible interval of 29-63. Among these, 26 are common, and 20 are specific.
To evaluate the model's goodness of fit, we run the Gibbs sampler for APAFA and TETRIS a second time, excluding 30 observations and imputing the missing values using the associated posterior predictive distributions. 
The mean squared errors were 0.0451 and 0.2833 for APAFA and TETRIS, respectively. Further details are reported in the Supplementary Material. 

\section{Discussion}
Motivated by the strict grouping imposed by MSFA approaches, we introduced a factor analytic model where the partition induced by specific latent factors is influenced by external information  but is not  strictly determined. While the primary interest has been in using a grouping categorical variable, consistent with the study structure of MSFA, the proposed APAFA is much more general and offers the possibility of including other categorical or continuous covariates. These act to promote sparsity in the latent factors, as done by \citet{schiavon2}, thereby  providing a natural  way to  aid in the modeling of the covariance structure and connecting with the modeling of heterogeneity in factor regression \citep{avalos2022heterogeneous}.

An important aspect of our contribution is that it is not  covariate dependence \textit{per se} that distinguishes specific from shared factors, but rather the sparsity caused by this dependence. Accordingly, there is no fundamental barrier to shared factors depending on common covariates, as long as such dependence does not lead to sparse latent factors. 
In principle, one could consider an alternative model in which shared factors vary with the covariates in a non-sparse manner. 
%
While exploring such an extension is conceptually feasible, it lies beyond the scope of this paper, which focuses on proposing an innovative alternative 
to the standard MSFA framework.  We leave these broader generalizations to future work.

The proposed solutions proved to be beneficial across many simulation settings, even in cases of group misspecification and complex  heterogeneity, both between and within the fixed groups.
Overlapping  group structures are nonetheless detectable and may differ in terms of flexibility and granularity. For example, in Section \ref{sec:birds}, study units are grouped together in scenarios where the number of studies closely matches the number of units. Conversely, in the example of Section \ref{sec:immune}, the model enables detecting higher granularity with respect to that specified \textit{a priori}, identifying clusters of a dozen units out of a thousand.

\bigskip
\begin{center}
{\large\bf SUPPLEMENTARY MATERIAL}
\end{center}
The supplementary material provides a detailed description of the Gibbs sampler algorithm used for posterior computation, along with a comprehensive overview of the simulation studies and real-data analyses. Specifically, it includes the data-generating mechanism for simulations, justification for prior hyperparameter selection and sensitivity analyzes, initialization strategies for MCMC chains, further simulation results, convergence and mixing assessment, post-processing steps, and runtime details. (pdf file).

\bigskip
\begin{center}
{\large\bf ACKNOWLEDGMENTS}\end{center}
The authors are grateful to Roberta De Vito, Giovanni Parmigiani, Lorenzo Schiavon and Alberto Tonolo for comments on an early version of this paper and to the Associate Editor and the Referees for their constructive comments and suggestions.
The authors report there are no competing interests to declare.

\bibliographystyle{agsm}
\bibliography{APAFA_biblio}

\end{document}

%% file: figure1.tex
\begin{tikzpicture}[scale=0.35]
	\definecolor{colorA}{rgb}{0.5, 0.0, 0.5}
	\definecolor{colorB}{rgb}{0.5, 0, 0}
	\definecolor{colorC}{rgb}{0, 0, 0.8}
	\definecolor{colorAA}{rgb}{0.45, 0, 0.55}
	\definecolor{colorAAA}{rgb}{0.65, 0, 0.65}
	\definecolor{colorAAAA}{rgb}{0.75, 0, 0.75}
	\definecolor{colorBB}{rgb}{0.55, 0, 0}
	\definecolor{colorCC}{rgb}{0, 0, 0.78}
	\definecolor{colorAAA}{rgb}{0.38, 0, 0.52}
	\definecolor{colorBBB}{rgb}{0.6, 0, 0}
	\definecolor{colorBBBB}{rgb}{0.7, 0, 0}
	\definecolor{colorCCC}{rgb}{0, 0, 0.85}
	\definecolor{colorCCCC}{rgb}{0, 0, 0.89}

	\definecolor{colorD}{rgb}{0.2, 0.2, 0.2}
	\definecolor{colorF}{rgb}{0.3, 0.3, 0.3}
	\definecolor{colorE}{rgb}{0.4, 0.4, 0.4}
	\definecolor{colorG}{rgb}{0.5, 0.5, 0.5}
	\definecolor{colorH}{rgb}{0.6, 0.6, 0.6}
	\definecolor{colorI}{rgb}{0.7, 0.7, 0.7}
	\definecolor{colorJ}{rgb}{0.8, 0.8, 0.8}
	\definecolor{colorK}{rgb}{0.9, 0.9, 0.9}
	\definecolor{color0}{rgb}{1, 1, 1}	
	
		\hspace{0cm}
	\hspace{-0cm}
	\fill[colorA] (0,0) rectangle ( 1,1);
	\fill[colorAA] (0,1) rectangle ( 1,2);
	\fill[colorAAA] (0,2) rectangle ( 1,3);
	\fill[colorAAAA] (0,3) rectangle ( 1,4);
	\fill[colorAA] (0,4) rectangle ( 1,5);
	\fill[colorC] (0,0-5) rectangle ( 1,1-5);
	\fill[colorBBB] (0,1-5) rectangle ( 1,2-5);
	\fill[colorB] (0,2-5) rectangle ( 1,3-5);
	\fill[colorBB] (0,3-5) rectangle ( 1,4-5);
	\fill[colorB] (0,4-5) rectangle ( 1,5-5);
	\fill[colorCC] (0,0-10) rectangle ( 1,1-10);
	\fill[colorCCC] (0,1-10) rectangle ( 1,2-10);
	\fill[colorC] (0,2-10) rectangle ( 1,3-10);
	\fill[colorCCCC] (0,3-10) rectangle ( 1,4-10);
	\fill[colorCCC] (0,4-10) rectangle (1,5-10);
	
	\fill[colorAA] (1,0) rectangle ( 2,1);
	\fill[colorA] (1,1) rectangle ( 2,2);
	\fill[colorAAAA] (1,2) rectangle ( 2,3);
	\fill[colorAAA] (1,3) rectangle ( 2,4);
	\fill[colorA] (1,4) rectangle ( 2,5);
	\fill[colorCC] (1,0-5) rectangle ( 2,1-5);
	\fill[colorBBBB] (1,1-5) rectangle ( 2,2-5);
	\fill[colorBB] (1,2-5) rectangle ( 2,3-5);
	\fill[colorB] (1,3-5) rectangle ( 2,4-5);
	\fill[colorBBB] (1,4-5) rectangle ( 2,5-5);
	\fill[colorCCC] (1,0-10) rectangle ( 2,1-10);
	\fill[colorCC] (1,1-10) rectangle (2,2-10);
	\fill[colorCCCC] (1,2-10) rectangle ( 2,3-10);
	\fill[colorCCC] (1,3-10) rectangle ( 2,4-10);
	\fill[colorC] (1,4-10) rectangle (2,5-10);
	
	\fill[colorAAA] (2,0) rectangle (3,1);
	\fill[colorAA] (2,1) rectangle (3,2);
	\fill[colorA] (2,2) rectangle (3,3);
	\fill[colorAAAA] (2,3) rectangle (3,4);
	\fill[colorA] (2,4) rectangle (3,5);
	\fill[colorC] (2,0-5) rectangle ( 3,1-5);
	\fill[colorB] (2,1-5) rectangle (3,2-5);
	\fill[colorBBB] (2,2-5) rectangle (3,3-5);
	\fill[colorBBBB] (2,3-5) rectangle (3,4-5);
	\fill[colorB] (2,4-5) rectangle ( 3,5-5);
	\fill[colorCCCC] (2,0-10) rectangle (3,1-10);
	\fill[colorC] (2,1-10) rectangle (3,2-10);
	\fill[colorCC] (2,2-10) rectangle ( 3,3-10);
	\fill[colorCCC] (2,3-10) rectangle ( 3,4-10);
	\fill[colorCCCC] (2,4-10) rectangle (3,5-10);
	
	\fill[colorAAAA] ( 3,0) rectangle (4,1);
	\fill[colorA] ( 3,1) rectangle (4,2);
	\fill[colorAA] ( 3,2) rectangle (4,3);
	\fill[colorAAA] ( 3,3) rectangle (4,4);
	\fill[colorAA] ( 3,4) rectangle (4,5);
	\fill[colorCCC] ( 3,0-5) rectangle ( 4,1-5);
	\fill[colorBB] ( 3,1-5) rectangle (4,2-5);
	\fill[colorBBBB] ( 3,2-5) rectangle (4,3-5);
	\fill[colorBBB] ( 3,3-5) rectangle (4,4-5);
	\fill[colorBB] ( 3,4-5) rectangle ( 4,5-5);
	\fill[colorCCC] ( 3,0-10) rectangle (4,1-10);
	\fill[colorCCCC] ( 3,1-10) rectangle (4,2-10);
	\fill[colorC] ( 3,2-10) rectangle ( 4,3-10);
	\fill[colorCC] ( 3,3-10) rectangle ( 4,4-10);
	\fill[colorCCC] (3,4-10) rectangle (4,5-10);
	
	\fill[colorAAA] ( 4,0) rectangle ( 5,1);
	\fill[colorAA] (  4,1) rectangle ( 5,2);
	\fill[colorAAAA] (  4,2) rectangle ( 5,3);
	\fill[colorA] (  4,3) rectangle ( 5,4);
	\fill[colorAAA] (  4,4) rectangle ( 5,5);
	\fill[colorC] (  4,0-5) rectangle (  5,1-5);
	\fill[colorB] (  4,1-5) rectangle ( 5,2-5);
	\fill[colorBBB] (  4,2-5) rectangle ( 5,3-5);
	\fill[colorBB] (  4,3-5) rectangle ( 5,4-5);
	\fill[colorBBB] (  4,4-5) rectangle (  5,5-5);
	\fill[colorCC] (  4,0-10) rectangle ( 5,1-10);
	\fill[colorC] (  4,1-10) rectangle ( 5,2-10);
	\fill[colorCCCC] (  4,2-10) rectangle (  5,3-10);
	\fill[colorCCC] (  4,3-10) rectangle (  5,4-10);
	\fill[colorCC] (4,4-10) rectangle (5,5-10);

	\hspace{-0cm}
	\fill[colorH] (11-5+1,0) rectangle (11-5+1+1,1);
	\fill[colorG] (11+1-5,1) rectangle (11-5+1+1,2);
	\fill[colorD] (11+1-5,2) rectangle (11-5+1+1,3);
	\fill[colorG] (11+1-5,3) rectangle (11-5+1+1,4);
	\fill[colorF] (11+1-5,4) rectangle (11-5+1+1,5);
	\fill[colorE] (11+1-5,0-5) rectangle (11-5+1+1,1-5);
	\fill[colorD] (11+1-5,1-5) rectangle (11-5+1+1,2-5);
	\fill[colorG] (11+1-5,2-5) rectangle (11-5+1+1,3-5);
	\fill[colorF] (11+1-5,3-5) rectangle (11-5+1+1,4-5);
	\fill[colorE] (11+1-5,4-5) rectangle (11-5+1+1,5-5);
	\fill[colorD] (11+1-5,0-10) rectangle (11-5+1+1,1-10);
	\fill[colorF] (11+1-5,1-10) rectangle (11-5+1+1,2-10);
	\fill[colorD] (11+1-5,2-10) rectangle (11-5+1+1,3-10);
	\fill[colorG] (11+1-5,3-10) rectangle (11-5+1+1,4-10);
	\fill[colorE] (11+1-5,4-10) rectangle (11+1-5+1,5-10);
	
	
	\fill[colorF] (12-4,0) rectangle (12-3,1);
	\fill[colorE] (12-4,1) rectangle (12-3,2);
	\fill[colorG] (12-4,2) rectangle (12-3,3);
	\fill[colorD] (12-4,3) rectangle (12-3,4);
	\fill[colorG] (12-4,4) rectangle (12-3,5);
	\fill[colorE] (12-4,-5) rectangle (12-3,1-5);
	\fill[colorG] (12-4,1-5) rectangle (12-3,2-5);
	\fill[colorF] (12-4,2-5) rectangle (12-3,3-5);
	\fill[colorD] (12-4,3-5) rectangle (12-3,4-5);
	\fill[colorG] (12-4,4-5) rectangle (12-3,5-5);
	
	\fill[colorE] (12-4,0-10) rectangle (12-3,1-10);
	\fill[colorG] (12-4,1-10) rectangle (12-3,2-10);
	\fill[colorD] (12-4,2-10) rectangle (12-3,3-10);
	\fill[colorF] (12-4,3-10) rectangle (12-3,4-10);
	\fill[colorD] (12-4,4-10) rectangle (12-3,5-10);
	
	\fill[colorG] (13-4,0) rectangle (13-3,1);
	\fill[colorD] (13-4,1) rectangle (13-3,2);
	\fill[colorE] (13-4,2) rectangle (13-3,3);
	\fill[colorG] (13-4,3) rectangle (13-3,4);
	\fill[colorF] (13-4,4) rectangle (13-3,5);
	\fill[colorE] (13-4,-5) rectangle (13-3,1-5);
	\fill[colorD] (13-4,1-5) rectangle (13-3,2-5);
	\fill[colorF] (13-4,2-5) rectangle (13-3,3-5);
	\fill[colorG] (13-4,3-5) rectangle (13-3,4-5);
	\fill[colorD] (13-4,4-5) rectangle (13-3,5-5);
	\fill[colorG] (13-4,0-10) rectangle (13-3,1-10);
	\fill[colorD] (13-4,1-10) rectangle (13-3,2-10);
	\fill[colorG] (13-4,2-10) rectangle (13-3,3-10);
	\fill[colorF] (13-4,3-10) rectangle (13-3,4-10);
	\fill[colorE] (13-4,4-10) rectangle (13-3,5-10);

	\hspace{-0.25cm}

	\fill[colorF] (11,0) rectangle (12,1);
	\fill[colorD] (12,0) rectangle (13,1);
	\fill[colorE] (13,0) rectangle (14,1);
	\fill[colorH] (15,0) rectangle (15,1);
	\fill[colorG] (15,0) rectangle (16,1);

	\fill[colorD] (11,1) rectangle (12,2);
	\fill[colorH] (12,1) rectangle (13,2);
	\fill[colorG] (13,1) rectangle (14,2);
	\fill[colorE] (14,1) rectangle (15,2);
	\fill[colorF] (15,1) rectangle (16,2);

	\fill[colorG] (11,2) rectangle (12,3);
	\fill[colorI] (12,2) rectangle (13,3);
	\fill[colorH] (13,2) rectangle (14,3);
	\fill[colorD] (14,2) rectangle (15,3);
	\fill[colorF] (15,2) rectangle (16,3);

	\hspace{0.0cm}
	\fill[colorAAAA] (19,0) rectangle (19+1,1);
	\fill[colorAAA] (19,1) rectangle (19+1,2);
	\fill[colorAAAA] (19,2) rectangle (19+1,3);
	\fill[colorAA] (19,3) rectangle (19+1,4);
	\fill[colorAAA] (19,4) rectangle (19+1,5);
	
	\fill[colorBB] (24-4,1-5) rectangle (24-3,2-5);
	\fill[colorBBB] (24-4,2-5) rectangle (24-3,3-5);
	\fill[colorBB] (24-4,3-5) rectangle (24-3,4-5);
	\fill[colorBBBB] (24-4,4-5) rectangle (24-3,5-5);
	
        \fill[colorC] (25-4,0-5) rectangle (25-3,1-5);
	\fill[colorCCC] (25-4,0-10) rectangle (25-3,1-10);
	\fill[colorCCC] (25-4,1-10) rectangle (25-3,2-10);
	\fill[colorCC] (25-4,2-10) rectangle (25-3,3-10);
	\fill[colorCCC] (25-4,3-10) rectangle (25-3,4-10);
	\fill[colorCCCC] (25-4,4-10) rectangle (25-3,5-10);

        \fill[colorCC] (25-3,0-5) rectangle (25-2,1-5);
	\fill[colorC] (25-3,0-10) rectangle (25-2,1-10);
	\fill[colorCCC] (25-3,1-10) rectangle (25-2,2-10);
	\fill[colorCCCC] (25-3,2-10) rectangle (25-2,3-10);
	\fill[colorC] (25-3,3-10) rectangle (25-2,4-10);
	\fill[colorC] (25-3,4-10) rectangle (25-2,5-10);

	
	\fill[colorCC] (24,0) rectangle (25,1);
	\fill[colorCCC] (25,0) rectangle (26,1);
	\fill[colorC] (26,0) rectangle (27,1);
	\fill[colorC] (27,0) rectangle (28,1);
        \fill[colorCCCC] (28,0) rectangle (29,1);

	\fill[colorC] (24,-1) rectangle (25,0);
	\fill[colorCCCC] (25,-1) rectangle (26,0);
	\fill[colorCC] (26,-1) rectangle (27,0);
	\fill[colorCCC] (27,-1) rectangle (28,0);
	\fill[colorC] (28,-1) rectangle (29,0);

	\fill[colorB] (24,1) rectangle (25,2);
	\fill[colorBB] (25,1) rectangle (26,2);
	\fill[colorB] (26,1) rectangle (27,2);
	\fill[colorBBB] (27,1) rectangle (28,2);
	\fill[colorBB] (28,1) rectangle (29,2);

	\fill[colorA] (24,2) rectangle (25,3);
	\fill[colorAA] (25,2) rectangle (26,3);
	\fill[colorAAA] (26,2) rectangle (27,3);
	\fill[colorA] (27,2) rectangle (28,3);
        \fill[colorA] (28,2) rectangle (29,3);

	\hspace{0cm}

	\hspace{3cm}  
	\fill[colorK] (23+8,0) rectangle (23+8+1,1);
	\fill[colorJ] (23+8,1) rectangle (23+8+1,2);
	\fill[colorH] (23+8,2) rectangle (23+8+1,3);
	\fill[colorJ] (23+8,3) rectangle (23+8+1,4);
	\fill[colorG] (23+8,4) rectangle (23+8+1,5);
	\fill[colorK] (23+8,-5) rectangle (23+8+1,1-5);
	\fill[colorJ] (23+8,1-5) rectangle (23+8+1,2-5);
	\fill[colorI] (23+8,2-5) rectangle (23+8+1,3-5);
	\fill[colorG] (23+8,3-5) rectangle (23+8+1,4-5);
	\fill[colorI] (23+8,4-5) rectangle (23+8+1,5-5);
	\fill[colorJ] (23+8,0-10) rectangle (23+8+1,1-10);
	\fill[colorK] (23+8,1-10) rectangle (23+8+1,2-10);
	\fill[colorG] (23+8,2-10) rectangle (23+8+1,3-10);
	\fill[colorH] (23+8,3-10) rectangle (23+8+1,4-10);
	\fill[colorJ] (23+8,4-10) rectangle (23+8+1,5-10);
	\fill[colorJ] (23+9,0) rectangle (24+9,1);
	\fill[colorI] (23+9,1) rectangle (24+9,2);
	\fill[colorH] (23+9,2) rectangle (24+9,3);
	\fill[colorJ] (23+9,3) rectangle (24+9,4);
	\fill[colorH] (23+9,4) rectangle (24+9,5);
	\fill[colorK] (23+9,-5) rectangle (24+9,1-5);
	\fill[colorI] (23+9,1-5) rectangle (24+9,2-5);
	\fill[colorJ] (23+9,2-5) rectangle (24+9,3-5);
	\fill[colorK] (23+9,3-5) rectangle (24+9,4-5);
	\fill[colorJ] (23+9,4-5) rectangle (24+9,5-5);
	\fill[colorH] (23+9,-10) rectangle (24+9,1-10);
	\fill[colorK] (23+9,1-10) rectangle (24+9,2-10);
	\fill[colorJ] (23+9,2-10) rectangle (24+9,3-10);
	\fill[colorI] (23+9,3-10) rectangle (24+9,4-10);
	\fill[colorK] (23+9,4-10) rectangle (24+9,5-10);
	
	\fill[colorJ] (25+9,0-10) rectangle (24+9,1-10);
	\fill[colorI] (25+9,1-10) rectangle (24+9,2-10);
	\fill[colorK] (25+9,2-10) rectangle (24+9,3-10);
	\fill[colorH] (25+9,3-10) rectangle (24+9,4-10);
	\fill[colorJ] (25+9,4-10) rectangle (24+9,5-10);
	\fill[colorH] (25+9,0) rectangle (24+9,1);
	\fill[colorK] (25+9,1) rectangle (24+9,2);
	\fill[colorI] (25+9,2) rectangle (24+9,3);
	\fill[colorG] (25+9,3) rectangle (24+9,4);
	\fill[colorH] (25+9,4) rectangle (24+9,5);
	\fill[colorK] (25+9,-5) rectangle (24+9,1-5);
	\fill[colorI] (25+9,1-5) rectangle (24+9,2-5);
	\fill[colorG] (25+9,2-5) rectangle (24+9,3-5);
	\fill[colorJ] (25+9,3-5) rectangle (24+9,4-5);
	\fill[colorK] (25+9,4-5) rectangle (24+9,5-5);
	\fill[colorH] (25+9,0-10) rectangle (24+9,1-10);
	\fill[colorJ] (25+9,1-10) rectangle (24+9,2-10);
	\fill[colorK] (25+9,2-10) rectangle (24+9,3-10);
	\fill[colorJ] (25+9,3-10) rectangle (24+9,4-10);
	\fill[colorI] (25+9,4-10) rectangle (24+9,5-10);
	\fill[colorJ] (25+9,0) rectangle (26+9,1);
	\fill[colorK] (25+9,1) rectangle (26+9,2);
	\fill[colorJ] (25+9,2) rectangle (26+9,3);
	\fill[colorH] (25+9,3) rectangle (26+9,4);
	\fill[colorK] (25+9,4) rectangle (26+9,5);
	\fill[colorK] (25+9,0-5) rectangle (26+9,1-5);
	\fill[colorJ] (25+9,1-5) rectangle (26+9,2-5);
	\fill[colorG] (25+9,2-5) rectangle (26+9,3-5);
	\fill[colorH] (25+9,3-5) rectangle (26+9,4-5);
	\fill[colorJ] (25+9,4-5) rectangle (26+9,5-5);
	\fill[colorK] (25+9,0-10) rectangle (26+9,1-10);
	\fill[colorJ] (25+9,1-10) rectangle (26+9,2-10);
	\fill[colorI] (25+9,2-10) rectangle (26+9,3-10);
	\fill[colorH] (25+9,3-10) rectangle (26+9,4-10);
	\fill[colorK] (25+9,4-10) rectangle (26+9,5-10);
	\fill[colorH] (26+9,0) rectangle (27+9,1);
	\fill[colorJ] (26+9,1) rectangle (27+9,2);
	\fill[colorH] (26+9,2) rectangle (27+9,3);
	\fill[colorK] (26+9,3) rectangle (27+9,4);
	\fill[colorJ] (26+9,4) rectangle (27+9,5);
	\fill[colorG] (26+9,0-5) rectangle (27+9,1-5);
	\fill[colorK] (26+9,1-5) rectangle (27+9,2-5);
	\fill[colorJ] (26+9,2-5) rectangle (27+9,3-5);
	\fill[colorK] (26+9,3-5) rectangle (27+9,4-5);
	\fill[colorH] (26+9,4-5) rectangle (27+9,5-5);
	\fill[colorI] (26+9,0-10) rectangle (27+9,1-10);
	\fill[colorH] (26+9,1-10) rectangle (27+9,2-10);
	\fill[colorI] (26+9,2-10) rectangle (27+9,3-10);
	\fill[colorK] (26+9,3-10) rectangle (27+9,4-10);
	\fill[colorI] (26+9,4-10) rectangle (27+9,5-10);
	
	\node at (-5.5,2) {\large \color{black} $\mathbf{Y}$};
	\node at (0.5,2) {\large \color{black} $\mathbf{H}$};
	\node at (6,2) {\large \color{black} $\mathbf{\Lambda^\top}$};
	\node at (13,2) {\large \color{black} $\mathbf{\Phi}$};
	\node at (18,2) {\large \color{black} $\mathbf{\Gamma}^\top$};
	\node at (25.5,2) {\large \color{black} $\mathbf{\epsilon}$};
	\node at (-1.8,2) {\Large $=$};
	\node at (10,2) {\Large $+$};
	\node at (22,2) {\Large $+$};
\end{tikzpicture}

%% file: neuron1.tex
\tikzset{%
  input/.style={
      circle,
      draw,
      color={rgb, 255:red, 107; green, 65; blue, 144 },
      draw opacity=1,
      fill={rgb, 255:red, 107; green, 65; blue, 144 },
      fill opacity=0.48,
      minimum size=0.5cm
    },
    every neuron/.style={
      circle,
      draw,
      color={rgb, 255:red, 65; green, 117; blue, 5 },
      draw opacity=1,
      fill={rgb, 255:red, 65; green, 117; blue, 5 },
      fill opacity=0.48,
      minimum size=0.5cm
    },
    neuron missing/.style={
      draw=none, 
      fill=none,
      scale=1,
      text height=0.333cm,
      execute at begin node=\color{black}$\vdots$
    },
    weights/.style={
      draw=none, 
      fill=none,
      scale=1,
      text height=0.333cm,
      execute at begin node=\color{black}$\bw$
    },
    output/.style={
      circle,
      draw,
      color={rgb, 255:red, 107; green, 65; blue, 144 },
      draw opacity=1,
      fill={rgb, 255:red, 107; green, 65; blue, 144 },
      fill opacity=0.48,
      minimum size=0.5cm
    },
}

\begin{tikzpicture}
    \definecolor{inputcolor}{RGB}{199, 199, 199} 
    \definecolor{hiddencolor}{RGB}{252, 252, 252} %

    \foreach \i in {1,2,3} {
        \node[draw, shape=rectangle, fill=inputcolor, minimum size=20pt] (x\i) at (0, -\i) {$x_{\i}$};    }
    \node at (0, -4) {$\vdots$}; 
    \node[draw, shape=rectangle, fill=inputcolor, minimum size=20pt] (x6) at (0, -5) {$x_{S}$};
    
    \foreach \i in {1,2} {
        \node[draw, circle, fill=hiddencolor, minimum size=20pt] (h\i) at (3, -\i-0.5) {$\varphi_{\i}$};
    }
    \node at (3, -3.5) {$\vdots$}; 
    \node[draw, circle, fill=hiddencolor, minimum size=20pt] (h6) at (3, -4.5) {$\varphi_{k}$};

    \foreach \i in {1,2,3,4} {
        \node[draw, circle, fill=inputcolor, minimum size=20pt] (y\i) at (6, -\i+0.5) {$y_{\i}$};
    }
    \node at (6, -4.5) {$\vdots$}; 
    \node[draw, circle, fill=inputcolor, minimum size=20pt] (y6) at (6, -5.5) {$y_{p}$};

    \foreach \i in {1} {
        \foreach \j in {1,2} {
            \draw[->] (x\i) -- (h\j) node[midway, above, sloped] {$\beta_{\i\j}$};
        }
        \draw[->] (x\i) -- (h6) node[midway, above, sloped] {$\beta_{\i k}$};
    }
    \foreach \i in {2,3} {
        \foreach \j in {1,2} {
            \draw[->] (x\i) -- (h\j) node[midway, above, sloped] {};
        }
        \draw[->] (x\i) -- (h6) node[midway, above, sloped] {$\beta_{\i k}$};
    }
    
    \draw[->] (x6) -- (h1) node[midway, above, sloped] {$\beta_{S1}$};
     \draw[->] (x6) -- (h2) node[midway, above, sloped] {$\beta_{S2}$};
    \draw[->] (x6) -- (h2) node[midway, above, sloped] {};
    \draw[->] (x6) -- (h6) node[midway, above, sloped] {$\beta_{Sk}$};

    \foreach \i in {1,2} {
        \foreach \j in {1} {
        \draw[->] (h\i) -- (y\j) node[midway, above, sloped] {$\gamma_{\i\j}$};
        }
    \draw[->] (h\i) -- (y6) node[midway, above, sloped] {$\gamma_{\i p}$};
    }
    \foreach \i in {1,2} {
        \foreach \j in {2,3,4} {
        \draw[->] (h\i) -- (y\j) node[midway, above, sloped] {};
    }
    \draw[->] (h\i) -- (y6) node[midway, above, sloped] {$\gamma_{\i p}$};
    }
    \draw[->] (h6) -- (y1) node[midway, above, sloped] {$\gamma_{k1}$};
    \draw[->] (h6) -- (y2) node[midway, above, sloped] {$\gamma_{k2}$};
    \draw[->] (h6) -- (y3) node[midway, above, sloped] {};
    \draw[->] (h6) -- (y4) node[midway, above, sloped] {};
    \draw[->] (h6) -- (y6) node[midway, above, sloped] {$\gamma_{kp}$};

\end{tikzpicture}

%% file: scenarios.tex
\hspace{2cm}
 	\begin{minipage}[t]{.5\textwidth}
  \hspace{0.8cm}
		\begin{tikzpicture}[scale=0.45]	 
			\definecolor{blue}{rgb}{0.25, 0.25, 0.25}
			\definecolor{red}{rgb}{0.9, 0.9, 0.9}
			\fill[red] (0,0) rectangle (1,6);
			\fill[red] (0,6) rectangle (1,12);
			\fill[blue] (0,12) rectangle (1,18);
			\fill[red] (1,0) rectangle (2,6);
			\fill[blue] (1,6) rectangle (2,12);
			\fill[red] (1,12) rectangle (2,18);
			\fill[blue] (2,0) rectangle (3,6);
			\fill[red] (2,6) rectangle (3,12);
			\fill[red] (2,12) rectangle (3,18);
			\fill[red] (4,0) rectangle (7,18);
			\fill[red] (8+0,0) rectangle (8+1,6);
			\fill[red] (8+0,6) rectangle (8+1,12);
			\fill[blue] (8+0,12) rectangle (8+1,18);
			\fill[red] (8+1,0) rectangle (8+2,6);
			\fill[blue] (8+1,6) rectangle (8+2,12);
			\fill[red] (8+1,12) rectangle (8+2,18);
			\fill[blue] (8+2,0) rectangle (8+3,6);
			\fill[blue] (8+2,6) rectangle (8+3,12);
			\fill[red] (8+2,12) rectangle (8+3,18);
			\fill[red] (12+0,0) rectangle (12+1,3);
			\fill[red] (12+0,3) rectangle (12+1,6);
			\fill[red] (12+0,6)  rectangle (12+1,12);
			\fill[blue] (12+0,12) rectangle (12+1,15);
			\fill[red] (12+0,15) rectangle (12+1,18);
			\fill[red] (12+1,0) rectangle (12+2,6);
			\fill[blue] (12+1,6) rectangle (12+2,9);
			\fill[red] (12+1,9) rectangle (12+2,12);
			\fill[red] (12+1,12) rectangle (12+2,18);
			\fill[blue] (12+2,0) rectangle (12+3,3);
			\fill[red] (12+2,3) rectangle (12+3,6);
			\fill[red] (12+2,6) rectangle (12+3,12);
			\fill[red] (12+2,12) rectangle (12+3,18);
   
			\node at (1.,-2) {$\textbf{(A,\:}$};
			\node at (2.5,-2) {$\,\textbf{A}^*\textbf{)}$};
			\node at (5.5,-2) {$\textbf{(B)}$};
			\node at (9.5,-2) {$\textbf{(C)}$};
			\node at (13.5,-2) {$\textbf{(D)}$};
		\end{tikzpicture}
	\end{minipage}%
\begin{minipage}[t]{.5\textwidth} \hspace{2.6cm}
\vspace{-9.35cm}\\ \hspace{3cm}
   	 \scalebox{-1}[1]{ \includegraphics[scale=0.273]{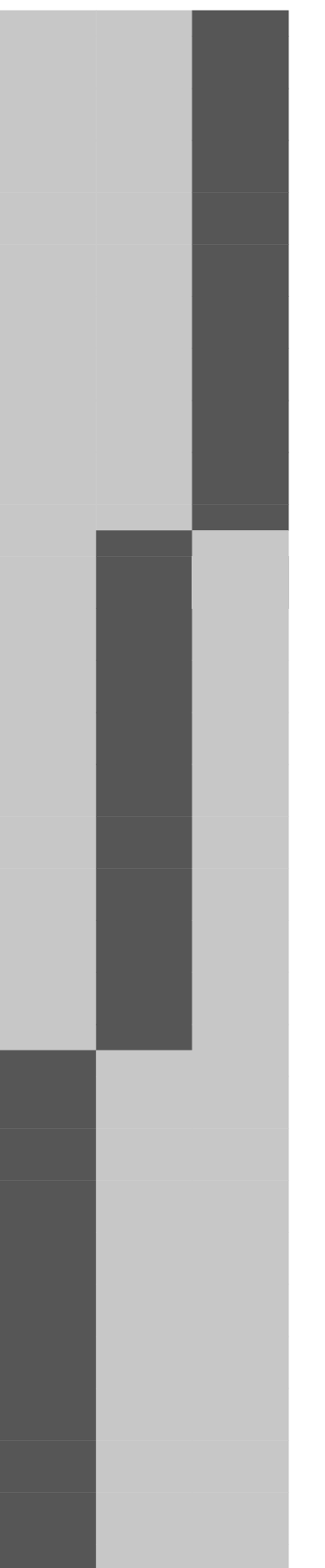}}\hspace{0.016cm}
   \scalebox{-1}[1]{  \includegraphics[scale=0.273]{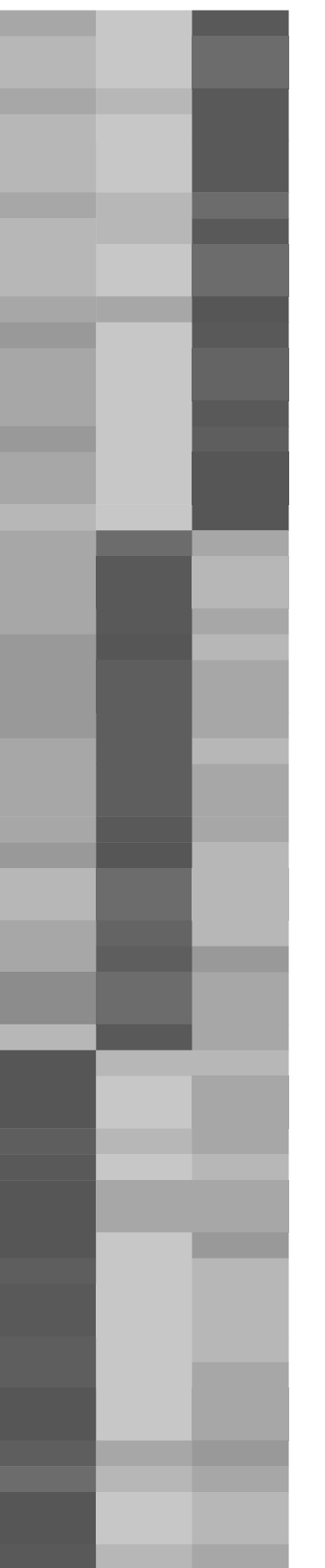}}\hspace{0.016cm}
 \rotatebox[origin=c]{180}{ \scalebox{-1}[1]{   \includegraphics[scale=0.273]{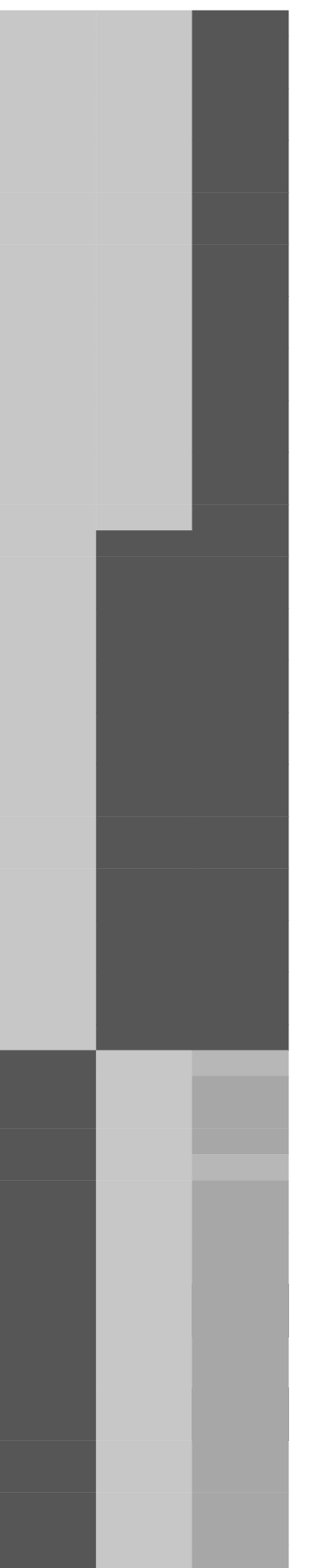}}}\hspace{0.016cm}
    \rotatebox[origin=c]{180}{\scalebox{-1}[1]{ \includegraphics[scale=0.273]{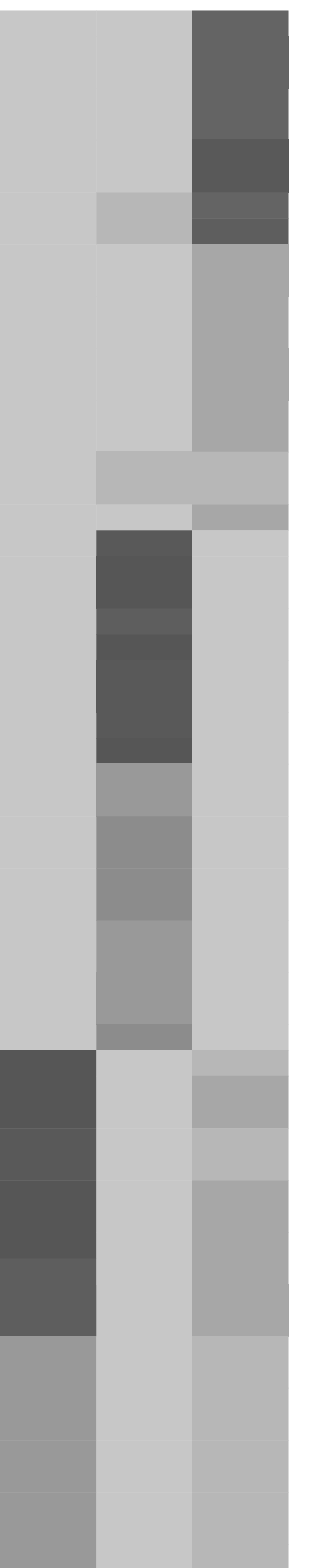}}}\\
\hspace{1.8cm} $\textbf{\;\;\;\;(A)}$\hspace{1cm}
 $\textbf{(A}^*$$\textbf{)}$ \hspace{1cm} $\textbf{(C)}$ \hspace{1cm} $\textbf{(D)}$\end{minipage}

%% file: tables.tex
\begin{table}[h!]
 \caption{Monte Carlo average (and interquartile range) of the posterior mean number of  factors and RV coefficients for $\Omega_1$, $\Omega_2$, and $\Omega_3$ on several simulation scenarios. Configuration $n>p$.}
 \begin{center}
	\begin{tabular}{llcccccc}
	 	Scen & Method & $d$ & $k$ &$\Omega_1$ & $\Omega_2$ & $\Omega_3$\\ 
		\hline
        A & \model& 3.00 (1.00)& 3.01 (0.23) & 0.90 (0.07) & 0.88 (0.07)  & 0.89 (0.05)  \\ 
		&TETRIS &3.67 (1.26) &4.91 (0.89) & 0.92 (0.07) & 0.93 (0.09) & 0.88  (0.08)\\ 
		\hline 
			$\text{A}^*$ &\model & 4.00 (1.00) & 3.00  (1.00) & 0.69 (0.13) & 0.78 (0.08) & 0.75 (0.08)  \\ 
		\hline 	 
		B&\model 	& 3.00 (0.00) & 0.00 (0.00) & 0.94 (0.04) & 0.94 (0.04) & 0.94 (0.04)  \\ 
	& TETRIS & 3.00 (0.00)& 0.00  (0.00) & 0.90 (0.12) & 0.92 (0.05)& 0.92 (0.09)  \\ 
		\hline 
			C&\model & 3.00 (0.00) & 3.00 (0.05)& 0.89 (0.05) & 0.78 (0.04) & 0.92 (0.02) \\ 
		 &TETRIS & 3.00 (1.00) & 2.00 (1.01) & 0.72 (0.09) & 0.75 (0.08) & 0.79 (0.05)  \\ 
		\hline 
		D&\model &3.00 (0.00) & 3.00 (0.12) & 0.91 (0.03)& 0.88 (0.03) & 0.90 (0.04) \\ 
	\end{tabular}
 \end{center}
	\label{table:scen_n>p}
\end{table}
\begin{table}[h!]
  \caption{Monte Carlo average (and interquartile range) of the posterior mean number of factors and RV coefficients for $\Omega_1$, $\Omega_2$, and $\Omega_3$ on several simulation scenarios. Configuration $n<p$.}
  \begin{center}
	\begin{tabular}{llccccc}
	 	Scen. & Method & $d$ & $k$ &$\Omega_1$ & $\Omega_2$ & $\Omega_3$  \\ 
		\hline
		A & \model&  5.00 (0.00) & 4.00 (0.00) & 0.91 (0.10) & 0.82 (0.1) & 0.58 (0.19)  \\ 
&TETRIS  & 2.41 (0.99) & 0.64 (0.86) &  0.80  (0.09) & 0.83 (0.15) & 0.86 (0.18)   \\ 
   \hline
    $\text{A}^*$ & \model&   5.99  (0.01) & 5.97 (0.04) & 0.87 (0.06) & 0.88 (0.06)& 0.86 (0.07) \\ 
    \hline
 B  & \model& 5.00 (0.00) & 0.00 (1.00) & 0.93 (0.06) & 0.93 (0.06) & 0.93 (0.06) \\ 
  & TETRIS & 3.00 (0.00) & 0.00  (0.00)  & 0.86 (0.08) & 0.90 (0.05) & 0.86 (0.12)   \\ 
    \hline
C& \model& 3.00  (0.75)& 3.00 (0.17)& 0.87 (0.07) & 0.78 (0.05) & 0.92 (0.03)  \\ 
&TETRIS & 5.00 (2.00) & 1.00 (1.95)& 0.71 (0.08) & 0.72 (0.13) & 0.76 (0.09)  \\ 
   \hline
D& \model& 5.00 (0.00) & 3.31 (1.00) & 0.89 (0.02) & 0.91 (0.03) & 0.88 (0.04) \\ 
  \hline
\end{tabular}
\end{center}\label{table:scen_n<p}
\end{table}

%% file: sample.bib
@article{devito1,
  title={Multi-study factor analysis},
  author={De Vito, Roberta and Bellio, Ruggero and Trippa, Lorenzo and Parmigiani, Giovanni},
  journal={Biometrics},
  volume={75},
  number={1},
  pages={337--346},
  year={2019},
  publisher={Wiley Online Library}
}

@article{ganzfried2013curatedovariandata,
  title={curatedOvarianData: clinically annotated data for the ovarian cancer transcriptome},
  author={Ganzfried, Benjamin Frederick and Riester, Markus and Haibe-Kains, Benjamin and Risch, Thomas and Tyekucheva, Svitlana and Jazic, Ina and Wang, Xin Victoria and Ahmadifar, Mahnaz and Birrer, Michael J and Parmigiani, Giovanni and others},
  journal={Database},
  volume={2013},
  year={2013},
  publisher={Oxford Academic}
}

@article{chandra2024sufa,
author = {Chandra, Noirrit Kiran and Dunson, David B. and  Xu, Jason},
title = {Inferring Covariance Structure from Multiple Data Sources via Subspace Factor Analysis},
journal = {Journal of the American Statistical Association},
volume = {0},
number = {ja},
pages = {1--25},
year = {2024},
publisher = {ASA Website},
doi = {10.1080/01621459.2024.2408777}

 

}

@article{poworoznek2021efficiently,
  title={Efficiently resolving rotational ambiguity in {B}ayesian matrix sampling with matching},
  author={Poworoznek, Evan and Ferrari, Federico and Dunson, David},
  journal={arXiv preprint arXiv:2107.13783},
  year={2021}
}

@article{xu2023identifiable,
  title={Identifiable and interpretable nonparametric factor analysis},
  author={Xu, Maoran and Herring, Amy H and Dunson, David B},
  journal={arXiv:2311.08254},
  year={2023}
}

@article{fruhwirth2024sparse,
  title={Sparse {B}ayesian factor analysis when the number of factors is unknown},
  author={Fr{\"u}hwirth-Schnatter, Sylvia and Hosszejni, Darjus and Lopes, Hedibert Freitas},
  journal={{B}ayesian Analysis},
  volume={1},
  number={1},
  pages={1--31},
  year={2024},
  publisher={International Society for {B}ayesian Analysis}
}

@article{papastamoulis2022identifiability,
  title={On the identifiability of {B}ayesian factor analytic models},
  author={Papastamoulis, Panagiotis and Ntzoufras, Ioannis},
  journal={Statistics and Computing},
  volume={32},
  number={2},
  pages={23},
  year={2022},
  publisher={Springer}
}

@article{devito2,
  title={{B}ayesian multistudy factor analysis for high-throughput biological data},
  author={De Vito, Roberta and Bellio, Ruggero and Trippa, Lorenzo and Parmigiani, Giovanni},
  journal={The Annals of Applied Statistics},
  volume={15},
  number={4},
  pages={1723--1741},
  year={2021},
  publisher={Institute of Mathematical Statistics}
}

@article{wen2016learning,
  title={Learning structured sparsity in deep neural networks},
  author={Wen, Wei and Wu, Chunpeng and Wang, Yandan and Chen, Yiran and Li, Hai},
  journal={Advances in neural information processing systems},
  volume={29},
  year={2016}
}

@article{jantre2023comprehensive,
  title={A comprehensive study of spike and slab shrinkage priors for structurally sparse {B}ayesian neural networks},
  author={Jantre, Sanket and Bhattacharya, Shrijita and Maiti, Tapabrata},
  journal={arXiv:2308.09104},
  year={2023}
}

@article{fortuin2022priors,
  title={Priors in {B}ayesian deep learning: a review},
  author={Fortuin, Vincent},
  journal={International Statistical Review},
  volume={90},
  number={3},
  pages={563--591},
  year={2022},
  publisher={Wiley Online Library}
}

@article{cui2022informative,
  title={Informative {B}ayesian neural network priors for weak signals},
  author={Cui, Tianyu and Havulinna, Aki and Marttinen, Pekka and Kaski, Samuel},
  journal={{B}ayesian Analysis},
  volume={17},
  number={4},
  pages={1121--1151},
  year={2022},
  publisher={International Society for {B}ayesian Analysis}
}

@article{sell2023trace,
  title={Trace-class {G}aussian priors for {B}ayesian learning of neural networks with {MCMC}},
  author={Sell, Torben and Singh, Sumeetpal Sidhu},
  journal={Journal of the Royal Statistical Society Series B: Statistical Methodology},
  volume={85},
  number={1},
  pages={46--66},
  year={2023},
  publisher={Oxford University Press US}
}

@article{schiavon1,
  title={Generalized infinite factorization models},
  author={Schiavon, Lorenzo and Canale, Antonio and Dunson, David B},
  journal={Biometrika},
  volume={109},
  number={3},
  pages={817--835},
  year={2022},
  publisher={Oxford University Press}
}

@article{schiavon2,
  title={Accelerated structured matrix factorization},
  author={Schiavon, Lorenzo and Nipoti, Bernardo and Canale, Antonio},
  journal={Journal of Computational and Graphical Statistics},
  year={2024}
}

@article{legramanti,
title={{B}ayesian cumulative shrinkage for infinite factorizations},
author={Legramanti, Sirio and Durante, Daniele and Dunson, David B},
journal={Biometrika},
volume={107},
number={3},
pages={745--752},
year={2020},
publisher={Oxford University Press}
}

@article{abdi2007rv,
  title={{RV} coefficient and congruence coefficient},
  author={Abdi, Herv{\'e}},
  journal={Encyclopedia of measurement and statistics},
  volume={849},
  pages={853},
  year={2007},
  publisher={Sage Thousand Oaks, CA}
}

@article{grabski2023bayesian,
  title={Bayesian combinatorial MultiStudy factor analysis},
  author={Grabski, Isabella N and De Vito, Roberta and Trippa, Lorenzo and Parmigiani, Giovanni},
  journal={The Annals of Applied Statistics},
  volume={17},
  number={3},
  pages={2212},
  year={2023},
  publisher={NIH Public Access}
}

@article{devito2023multi,
  title={Multi-study factor regression model: an application in nutritional epidemiology},
  author={De Vito, Roberta and Avalos-Pacheco, Alejandra},
  journal={arXiv preprint arXiv:2304.13077},
  year={2023}
}

@article{avalos2022heterogeneous,
  title={Heterogeneous large datasets integration using {B}ayesian factor regression},
  author={Avalos-Pacheco, Alejandra and Rossell, David and Savage, Richard S},
  journal={Bayesian Analysis},
  volume={17},
  number={1},
  pages={33--66},
  year={2022},
  publisher={International Society for {B}ayesian Analysis}
}

@article{fruhwirth2023generalized,
  title={Generalized cumulative shrinkage process priors with applications to sparse {B}ayesian factor analysis},
  author={Fr{\"u}hwirth-Schnatter, Sylvia},
  journal={Philosophical Transactions of the Royal Society A},
  volume={381},
  number={2247},
  pages={20220148},
  year={2023},
  publisher={The Royal Society}
}

@article{bhattacharya2011sparse,
  title={Sparse {B}ayesian infinite factor models},
  author={Bhattacharya, Anirban and Dunson, David B},
  journal={Biometrika},
  volume={98},
  number={2},
  pages={291--306},
  year={2011},
  publisher={Oxford University Press}
}

@book{national2019reproducibility,
	title={Reproducibility and replicability in science},
	author={{National Academies of Sciences} and {Policy and Global Affairs} and {Board on Research Data and Information} and {Division on Engineering and Physical Sciences} and {Committee on Applied and Theoretical Statistics} and {Board on Mathematical Sciences} and others},
	year={2019},
	publisher={National Academies Press}
}

@article{irizarry2003exploration,
	title={Exploration, normalization, and summaries of high density oligonucleotide array probe level data},
	author={Irizarry, Rafael A and Hobbs, Bridget and Collin, Francois and Beazer-Barclay, Yasmin D and Antonellis, Kristen J and Scherf, Uwe and Speed, Terence P},
	journal={Biostatistics},
	volume={4},
	number={2},
	pages={249--264},
	year={2003},
	publisher={Oxford University Press}
}

@article{shi2006microarray,
	title={The {M}icroArray {Q}uality {C}ontrol ({MAQC}) project shows inter-and intraplatform reproducibility of gene expression measurements.},
	author={Shi, Leming and Reid, Laura H and Jones, Wendell D and Shippy, Richard and Warrington, Janet A and Baker, Shawn C and Collins, Patrick J and De Longueville, Francoise and Kawasaki, Ernest S and Lee, Kathleen Y and others},
	journal={Nature biotechnology},
	volume={24},
	number={9},
	pages={1151--1161},
	year={2006}
}

@article{ovaskainen2017make,
	title={How to make more out of community data? {A} conceptual framework and its implementation as models and software},
	author={Ovaskainen, Otso and Tikhonov, Gleb and Norberg, Anna and Guillaume Blanchet, F and Duan, Leo and Dunson, David and Roslin, Tomas and Abrego, Nerea},
	journal={Ecology letters},
	volume={20},
	number={5},
	pages={561--576},
	year={2017},
	publisher={Wiley Online Library}
}

@article{ibp,
	title={The {I}ndian {B}uffet {P}rocess: an Introduction and Review.},
	author={Griffiths, Thomas L and Ghahramani, Zoubin},
	journal={Journal of Machine Learning Research},
	volume={12},
	number={4},
	year={2011}
}

@article{GriffinHoff,
	author = {Maryclare Griffin and Peter D. Hoff},
	title = {Structured Shrinkage Priors},
	journal = {Journal of Computational and Graphical Statistics},
	volume = {33},
	number = {1},
	pages = {1--14},
	year = {2024},
	publisher = {Taylor \& Francis},
	doi = {10.1080/10618600.2023.2233577}
}

@article{roy2021perturbed,
	title={Perturbed factor analysis: Accounting for group differences in exposure profiles},
	author={Roy, Arkaprava and Lavine, Isaac and Herring, Amy H and Dunson, David B},
	journal={The Annals of applied statistics},
	volume={15},
	number={3},
	pages={1386},
	year={2021},
	publisher={NIH Public Access}
}

@article{neal1990learning,
	title={Learning stochastic feedforward networks},
	author={Neal, Radford M},
	journal={Department of Computer Science, University of Toronto},
	volume={64},
	number={1283},
	pages={1577},
	year={1990},
	publisher={Citeseer}
}

@article{tang2013learning,
	title={Learning stochastic feedforward neural networks},
	author={Tang, Charlie and Salakhutdinov, Russ R},
	journal={Advances in Neural Information Processing Systems},
	volume={26},
	year={2013}
}

@inproceedings{kingma2014stochastic,
	title={Stochastic gradient {VB} and the variational auto-encoder},
	author={Kingma, Diederik P and Welling, Max},
	booktitle={Second international conference on learning representations, ICLR},
	volume={19},
	year={2014}
}

@incollection{sen2024bayesian,
	title={Bayesian neural networks and dimensionality reduction},
	author={Sen, Deborshee and Papamarkou, Theodore and Dunson, David},
	booktitle={Handbook of Bayesian, Fiducial, and Frequentist Inference},
	pages={188--209},
	year={2024},
	publisher={Chapman and Hall/CRC}
}

@incollection{goan2020bayesian,
	title={Bayesian neural networks: An introduction and survey},
	author={Goan, Ethan and Fookes, Clinton},
	booktitle={Case Studies in Applied {B}ayesian Data Science},
	pages={45--87},
	year={2020},
	publisher={Springer}
}

@article{arbel2023primer,
	title={A primer on {B}ayesian neural networks: review and debates},
	author={Arbel, Julyan and Pitas, Konstantinos and Vladimirova, Mariia and Fortuin, Vincent},
	journal={arXiv preprint arXiv:2309.16314},
	year={2023}
}

@inproceedings{nalisnick2019dropout,
	title={Dropout as a structured shrinkage prior},
	author={Nalisnick, Eric and Hern{\'a}ndez-Lobato, Jos{\'e} Miguel and Smyth, Padhraic},
	booktitle={Proceedings of the 36th International Conference on Machine Learning},
	pages={4712--4722},
	year={2019},
	organization={PMLR}
}

@article{lindstrom2015large,
  title={Large-scale monitoring of waders on their boreal and arctic breeding grounds in northern Europe},
  author={Lindstr{\"o}m, {\AA}ke and Green, Martin and Husby, Magne and K{\aa}l{\aa}s, John Atle and Lehikoinen, Aleksi},
  journal={Ardea},
  volume={103},
  number={1},
  pages={3--15},
  year={2015},
  publisher={BioOne}
}

@article{pena2023relationships,
  title={The relationships between urbanization and bird functional traits across the streetscape},
  author={Pena, Jo{\~a}o Carlos and Ovaskainen, Otso and MacGregor-Fors, Ian and Teixeira, Camila Palhares and Ribeiro, Milton Cezar},
  journal={Landscape and Urban Planning},
  volume={232},
  pages={104685},
  year={2023},
  publisher={Elsevier}
}

@article{hicks2018missing,
  title={Missing data and technical variability in single-cell {RNA}-sequencing experiments},
  author={Hicks, Stephanie C and Townes, F William and Teng, Mingxiang and Irizarry, Rafael A},
  journal={Biostatistics},
  volume={19},
  number={4},
  pages={562--578},
  year={2018},
  publisher={Oxford University Press}
}

@article{vallejos2017normalizing,
  title={Normalizing single-cell {RNA} sequencing data: challenges and opportunities},
  author={Vallejos, Catalina A and Risso, Davide and Scialdone, Antonio and Dudoit, Sandrine and Marioni, John C},
  journal={Nature methods},
  volume={14},
  number={6},
  pages={565--571},
  year={2017},
  publisher={Nature Publishing Group US New York}
}
